\documentclass[11pt]{article}

 \usepackage{fullpage}
\usepackage[margin=1in]{geometry}
\usepackage{amsmath, amsthm, amsfonts, thm-restate, tikz}
\usepackage{svg}

\usepackage[kw]{pseudo}
\usepackage{xspace}
\usepackage{hyperref}
\usepackage{mathtools}
\usepackage{subcaption}

\usetikzlibrary{shadows.blur, arrows, backgrounds}

\newcommand{\CLARKSON}{\textbf{Clarkson}\xspace}

\newcommand{\CLARKSONBB}{\textbf{ClarksonBB}\xspace}

\newcommand{\REFINE}{\textbf{GTRefine}\xspace}
\newcommand{\MERGE}{\textbf{GTMerge}\xspace}

\newcommand{\dist}{\mathbf{d}}
\newcommand{\ball}{\mathbf{B}}
\newcommand{\poly}{\mathrm{poly}}
\newcommand{\spread}{\Delta}
\newcommand{\R}{\mathbb{R}}
\newcommand{\pred}{\mathsf{pred}}
\newcommand{\maxkey}{\mathsf{maxkey}}
\newcommand{\key}{\mathsf{key}}
\newcommand{\e}{\varepsilon}
\newcommand{\emin}{\e_\mathsf{min}}
\newcommand{\vor}{\mathsf{V}}
\newcommand{\inradius}{\mathsf{inrad}}
\newcommand{\radius}{\mathsf{R}}
\newcommand{\outradius}{\radius}

\newcommand{\node}[2]{(#1, #2)}
\newcommand{\leaf}[1]{\mathsf{leaf}(#1)}
\newcommand{\noderad}[1]{\mathsf{noderad}(#1)}

\newcommand{\nodecenter}[1]{\mathsf{ctr}(#1)}

\newcommand{\points}{\mathsf{pts}}
\newcommand{\rmax}{\mathsf{R_{max}}}
\newcommand{\splitdist}{\mathsf{splitdist}}
\newcommand{\farthest}{\mathsf{farthest}}
\newcommand{\removemax}{\mathsf{removemax}}
\newcommand{\findmax}{\mathsf{max}}

\newcommand{\bucketsize}{c}            
\newcommand{\annulus}{\theta}               

\newcommand{\vorapx}{\beta}                 
\newcommand{\lazy}{\alpha}                  
\newcommand{\radapprox}{c}             
\newcommand{\newlazy}{\zeta}                

\newcommand{\scale}{\alpha}                 
\newcommand{\gpapprox}{\beta}               
\newcommand{\locallygreedy}{\gamma}         
\newcommand{\greedy}{\gamma}         

\newcommand{\aspect}{\tau}
\newcommand{\pack}{\pi}

\newcommand{\because}[1]{& \left[\text{#1}\right]}

\newtheorem{theorem}{Theorem}[section]
\newtheorem{lemma}[theorem]{Lemma}
\newtheorem{corollary}[theorem]{Corollary}

\title{Simple Construction of Greedy Trees and Greedy Permutations} 
 \author{
     Oliver A. Chubet\\
     North Carolina State University\\
     \texttt{oliver.chubet@gmail.com}
     \and
     Donald R. Sheehy\\
     North Carolina State University\\
     \url{https://donsheehy.net/}\\
     \texttt{don.r.sheehy@gmail.com}\\
     \texttt{https://orcid.org/0000-0002-9177-2713}
     \and
     Siddharth S. Sheth\\
     North Carolina State University\\
     \texttt{sheth.sid@gmail.com}
 }

\begin{document}
    \maketitle
    \begin{abstract}
  Greedy permutations, also known as Gonzalez Orderings or Farthest Point Traversals are a standard way to approximate $k$-center clustering and have many applications in sampling and approximating metric spaces.
  A greedy tree is an added structure on a greedy permutation that tracks the (approximate) nearest predecessor.
  Greedy trees have applications in proximity search as well as in topological data analysis.
  For metrics of doubling dimension $d$, a $2^{O(d)}n\log n$ time algorithm to compute a greedy permutation is known, but it is randomized and also, quite complicated.
  Its construction involves a series of intermediate structures and $O(n \log n)$ space.
  In this paper, we show how to construct greedy permutations and greedy trees using a simple variation of an algorithm of Clarkson that was shown to run in $2^{O(d)}n\log \spread$ time, where the spread $\spread$ is the ratio of largest to smallest pairwise distances.
  The improvement comes from the observation that the greedy tree can be constructed more easily than the greedy permutation.
  This leads to a linear time algorithm for merging two approximate greedy trees and a $2^{O(d)}n \log n$ time algorithm for computing the tree.
  Then, trivially, we can run this algorithm in $2^{O(d)}n$ parallel time.
  Then, we show how to extract a $(1+\frac{1}{n})$-approximate greedy permutation from the approximate greedy tree in the same asymptotic running time.
\end{abstract}


    \section{Introduction}\label{sec:introduction}

A greedy permutation of $n$ points in a metric space is an ordering that starts with any of the points and proceeds so that each point maximizes the minimum distance to its predecessors.
Greedy permutations first appeared as a heuristic for solving the traveling salesman problem~\cite{Rosenkrantz77analysis}.
It is a natural heuristic to produce metrically uniform samples.
(See Figure~\ref{fig:gp}.)
The first $k$ points in the greedy permutation give a $2$-approximation to the $k$-center clustering problem~\cite{Dyer85Simple,Gonzalez85clustering}.
This useful property has led to widespread use in clustering and other approximation algorithms.

\begin{figure}[ht]
    \centering
    \includegraphics[scale=0.65]{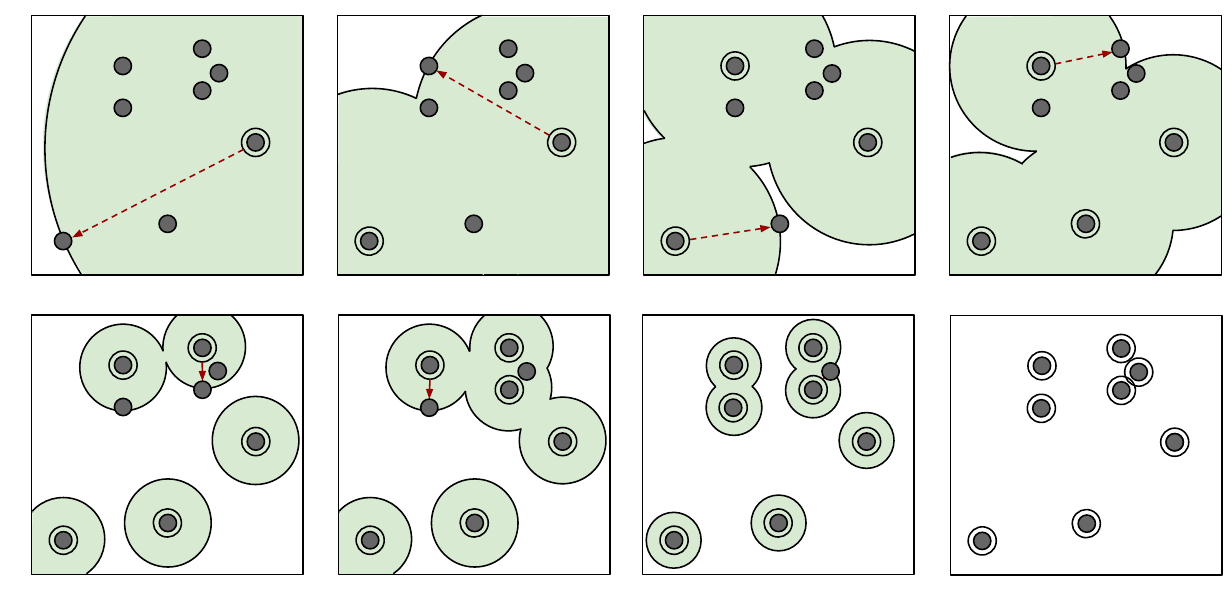}
    \caption{
        Each panel depicts a prefix of the greedy permutation and the corresponding net.
    }
    \label{fig:gp}
\end{figure}

Let $P = \{p_0,\ldots, p_{n-1}\}$ be a finite metric space.
Let the $i$th prefix $P_i:=\{p_0,\ldots, p_{i-1}\}$ be the first $i$ points.
Define the distance from a point $p$ to a set $S$ as $\dist(p,S):= \min_{s\in S}\dist(p,s)$.
Given an approximation factor $\gpapprox\ge 1$, the ordering is \textbf{$\gpapprox$-greedy} if for all $i\ge 1$, $\gpapprox \dist(p_i, P_i) \ge \max_{p\in P}\dist(p, P_i)$.
That is, each point is approximately the farthest next point.

Gonzalez~\cite{Gonzalez85clustering} gives the following simple algorithm for computing the greedy permutation in quadratic time.
Start with any point and assign it as the nearest neighbor of each of the remaining points.
At each step, the point that is farthest from its nearest neighbor is appended to the order.
One pass over the remaining points suffices to update their nearest neighbor, either it stays the same or is updated to be the most recently inserted point.
Partitioning a point set according to nearest neighbors as in Gonzalez's algorithm yields a finite Voronoi diagram.
In his work on nearest neighbor search in general metric spaces, Clarkson~\cite{clarkson02nearest} developed a more efficient way to maintain the finite Voronoi diagram.
His approach was to store a graph on the points added thus far so that updating nearest neighbors only requires checking local neighborhoods.
This added structure motivated the neighbor graph structure in the finite Voronoi diagrams of this paper.
We refer to this finite Voronoi approach to computing greedy permutations as Clarkson's algorithm (see Figure~\ref{fig:vor}).

\begin{figure}[ht]
\centering
\includegraphics[scale=0.65]{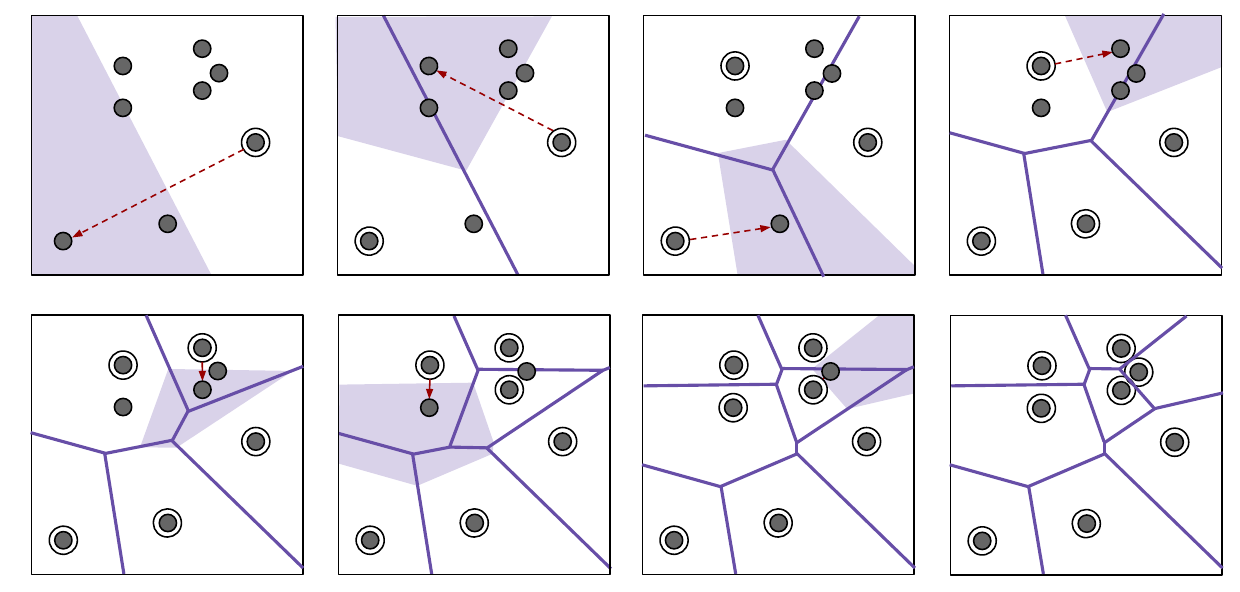}
\caption{ A greedy permutation is computed via incremental Voronoi construction.
    The shaded regions cover points that are closer to the new site than its previous nearest site.
    }
    \label{fig:vor}
\end{figure}

Har-Peled and Mendel~\cite{har-peled06fast} proved that Clarkson's algorithm runs in $O(n\log \spread)$ time for greedy permutations in doubling metrics.
Here, $\spread$ is the spread of the input, the ratio of the largest to smallest pairwise distances among the input points.
To reduce the $\log\spread$ term to $\log n$, Har-Peled and Mendel~\cite{har-peled06fast} gave a randomized algorithm that proceeds via a series of approximations.
It proceeds by first constructing a graph on $P$ with $O(n\log n)$ edges whose shortest paths give a $3n$ approximation to true distances.
This graph is constructed using random sampling and is then converted to a hierarchical separator tree (HST) that approximates the shortest paths up to a factor of $n$.
The resulting HST is then used to speed up Clarkson's algorithm by only revealing points by splitting nodes of the tree as the distance to the next point gets sufficiently small (the details of this step mirror the simplified version in Section~\ref{sec:trees_to_permutations}).
Suppressing (for now) constants that depend on the dimension, the expected running time is $O(n\log n)$ and the space usage is $O(n\log n)$ (for the spanner construction).
The result is a $(1+\frac{1}{n^2})$-greedy permutation.
The Har-Peled-Mendel Algorithm shows that it is possible to compute greedy permutations in $O(n\log n)$ time, but leaves open the question of computing them deterministically in linear space.
Moreover, the sequence of approximations is theoretically sound but is substantially more complicated than Clarkson's algorithm.

In this paper, we give simple algorithms for greedy permutations.
We use a simple proximity search data structure, the greedy tree, that can be encoded using a greedy ordering and a reference for each point to its (approximate) nearest predecessor.
This is defined formally in Section~\ref{sec:fvm_and_gts}.
The algorithms presented match the best-known asymptotic running times for constructing greedy permutations and greedy trees.

    \section{Related Work}
\label{sec:related}

In the $1970$'s several data structures were introduced to solve proximity search problems, including Voronoi diagrams~\cite{shamos75closest}, $kd$-trees~\cite{Bentley75,friedman1977algorithm}, and quadtrees~\cite{Finkel74}.
Hierarchical structures often have the benefit of logarithmic time search queries.
The introduction of ball trees~\cite{omohundro87efficient,omohundro89five}, $vp$-trees~\cite{Yianilos93}, $M$-trees~\cite{Ciaccia97}, and metric trees~\cite{uhlmann91satisfying} allowed these hierarchical methods to be extended to general metric settings.
These data structures are intuitive, yet few have strong theoretic guarantees for worst-case search query times.

To get strong guarantees in a non-Euclidean setting, the doubling dimension was introduced as a way to characterize metrics that resemble low-dimensional Euclidean spaces.
The \emph{doubling constant} of a metric space is the minimum number $\rho$ such that any ball can be covered by at most $\rho$ balls of half the radius and its \emph{doubling dimension} is $\log_2 \rho$.
Metric spaces with bounded doubling dimension are caleed \emph{doubling metrics}.
The following lemma bounds the cardinality of packed and bounded sets in doubling metrics~\cite{krauthgamer04navigating}.
\begin{lemma}\label{lem:std_packing}[Standard Packing Lemma]
    Let $(X, \dist)$ be a metric space with $\dim(X) = d$.
    If $Z \subseteq X$ is $r$-packed and can be covered by a metric ball of radius $R$ then $|Z| \le \left(\frac{4R}{r}\right)^d$.
\end{lemma}

Data structures such as the navigating nets of Krauthgamer and Lee~\cite{Krauthgamer04}, the cover trees of Beygelzimer et al.~\cite{beygelzimer06cover}, and the net-trees of Har-Peled and Mendel~\cite{har-peled06fast} all build on hierarchies of finer and finer nets in doubling metrics.
Har-Peled and Mendel~\cite{har-peled06fast} point out that to give a correct analysis in terms of the spread, one needs a model of computation that can handle inputs with large spread.
We use their asymptotic model which assumes that basic operations like arithmetic, logarithms, and floor functions can be computed on distances.
Moreover, we assume that distances can be computed in constant time.

A well-known case of an algorithm that eliminates the dependence on the spread is in the work of Callahan and Kosaraju~\cite{callahan95decomposing}.
That case applied in Euclidean settings.
This paper presents an approach for eliminating a $\log \spread$ term from one of the simplest greedy permutation algorithms.
Likely, this approach will also apply to other problems that seem to depend on the spread in general metric spaces.

\paragraph*{Contributions}
In this paper, we give a simple variation of Clarkson's algorithm that deterministically computes a greedy permutation in $O(n\log n)$ time and linear space.
Two main observations make this possible.
First, different permutations can produce the same tree.
As such, there is enough flexibility in the order in Clarkson's algorithm to eliminate the $\log n$ factor from the heap.
Second, using the last step of the Har-Peled-Mendel algorithm with a pair of greedy trees rather than an HST, one can merge greedy trees in linear time.
Recursive merging then gives an $O(n\log n)$ time algorithm to produce a constant factor approximate greedy tree.
From this, a $(1+\frac{1}{\poly(n)})$-approximate greedy permutation can be computed from the greedy tree in $O(n\log n)$ time.

    \section{Finite Voronoi Diagrams and Greedy Trees}\label{sec:fvm_and_gts}

In Section~\ref{sec:incremental}, we define finite Voronoi diagrams and detail a simple, local algorithm to construct them incrementally.
The algorithm is a simplification of an algorithm of Clarkson~\cite{clarkson99nearest,clarkson02nearest} and will form the basis for all of the following algorithms.\footnote{The main simplification here is that Clarkson's structure used a directed graph, but sometimes required traversing some edges in reverse.  We use an undirected graph.}
We show in Section~\ref{sec:simple_clarkson} how it can be used to efficiently compute greedy permutations with only a small amount of overhead compared to the quadratic algorithm.

In Section~\ref{sec:greedy_trees}, we define the greedy tree, a binary tree that represents both the greedy permutation as well as the nearest predecessor relationship.
The greedy tree is constructed at the same time as the finite Voronoi diagram.
It can be used on its own as a proximity search data structure~\cite{chubet23proximity} or as a way to speed up the point location work when computing greedy permutations.
We give a more detailed overview of the main results of the paper in Section~\ref{sec:main_results}.

\subsection{Incremental Construction of Finite Voronoi Diagrams}\label{sec:incremental}

Let $P$ be a finite metric space.
Let $S \subseteq P$ be any subset.
A $\vorapx$-approximate \emph{finite Voronoi diagram (FVD)} on $P$ with \emph{sites} $S$ is a partition of $P$ into sets called \emph{cells}.
There is a cell $\vor(p)$ for each site $p \in S$ satisfying the following.

\begin{quote}
    \textbf{The Cell Invariant:} If $p\in S$ and $p'\in \vor(p)$, then
      $\dist(p,p') \le \vorapx \min_{q\in S}\dist(p',q)$.
  \end{quote}

The \emph{out-radius} of a site $p$ is the distance from $p$ to its farthest point in $\vor(p)$.
The \emph{in-radius} of $p$ is the distance to the nearest point of $P$ that is not in $\vor(p)$.
Let $\outradius(p)$ and $\inradius(p)$ denote the out-radius and in-radius of $p$ respectively.
In Euclidean Voronoi diagrams, the out-radius is always greater than the in-radius, but in FVDs, this can be inverted.
For example, if $S = P$, then the out-radii are all zero and the inradii are all positive.

For a strict Voronoi diagram ($\vorapx = 1$), points move every time a new nearest neighbor is inserted.
The movement of points is called \emph{point location}.
Assume a site $p'\in \vor(p)$ is being inserted in a $\vorapx$-approximate Voronoi diagram.
A point $q'\in\vor(q)$ \emph{must} move into $V(p')$ when $\vorapx\dist(q',p') <  \dist(q', q)$.
Otherwise, the Cell Invariant would be violated.
However, if the new nearest neighbor is not significantly closer, then moving the point may not be unnecessary.
The \emph{lazy move} heuristic is a more conservative condition for point location that is critical to the theoretical results that follow.
\begin{quote}
    \textbf{$\lazy$-Lazy Move:}
    A point $q'\in \vor(q)$ does not move into $\vor(p')$ if $\lazy\dist(q',p') \ge \dist(q',q)$.
\end{quote}
In an $\lazy$-lazy move, the distance to the nearest neighbor has gone down by a factor of $\lazy$.
The lazy move constant $\lazy$ determines which points are \emph{forbidden} from moving and the approximation constant $\vorapx$ determines which points are \emph{required} to move.
In the simplest case, choosing $\lazy= \vorapx$ is appropriate.
Later, we will require $\lazy < \vorapx$.

In the quadratic algorithm for greedy permutations, the point location step requires linear time as every point is touched to see if it moves.
An improvement is possible if the point location touches can be limited to the local neighborhood of the newly inserted site.
This is accomplished by maintaining a graph on the sites called the \emph{neighbor graph} that plays a role analogous to the Delaunay triangulation for geometric Voronoi diagrams.
The neighbor graph satisfies the following invariant.

\begin{quote}
    \textbf{The Neighbor Invariant:} For all sites $a$ and $b$, if there exist points $a'\in \vor(a)$ and $b'\in \vor(b)$ such that $\dist(a',b')< \dist(b',b)$, then there is an edge from $a$ to $b$.
\end{quote}

Note that the Neighbor Invariant only provides a one-sided requirement.
It says which edges \emph{must} be present, but not which edges \emph{cannot} be present.
Throughout, the neighbor graph will contain more edges than are strictly necessary for the invariant.
This allows for faster updates and pruning of edges.

Given a point $p'\in \vor(p)$, inserting $p'$ as a new site has three steps.
\begin{enumerate}
    \item \textbf{Point Location}: Points are moved into the new cell.
    For all neighbors $q$ of $p$ and all points $q'\in \vor(q)$, if $\lazy\dist(q',p')\le \dist(q,q')$, then $q'$ is moved into $\vor(p')$.
    \item \textbf{Graph Update}: Neighbor graph edges are added incident to the new site.
    An edge from $p'$ is added to any site within two steps of $p$ in the neighbor graph.
    \item \textbf{Pruning}: Edges that are too long to be required by the Neighbor Invariant are removed.
    Among points whose out-radius has changed, edges $ab$ are removed if
    \[
        \dist(a,b)> \outradius(a) + \outradius(b) + \max\{\outradius(a), \outradius(b)\}.
    \]
\end{enumerate}

Note that in a $\vorapx$-approximate Voronoi diagram the insertion distance of $p'$ is a $\vorapx$-approximation of the distance from $p'$ to the sites.
The Neighbor Invariant guarantees that the point location step maintains the Cell Invariant.
The correctness of the graph update and pruning steps are also straightforward adaptations of prior work~\cite{clarkson99nearest}.
For completeness, we included their proofs in Appendix~\ref{sec:voronoi_lemmas}.
The proofs only require showing that these three steps preserve the Cell Invariant and the Neighbor Invariant.

\subsection{Clarkson's Algorithm for Greedy Permutations}\label{sec:simple_clarkson}

Using the incremental FVD algorithm above, it is now easy to describe the simplest version of Clarkson's algorithm for greedy permutations.
\begin{pseudo}*
    \hd{\CLARKSON}(P): \\
    Start with any point $p_0\in P$ and make a FVD with just one cell, $\vor(p_0)$.\\
    Initialize the greedy permutation $GP$ to be a list containing just $p_0$.\\
    Maintain a max-heap $H$ of the sites in which the key is the out-radius.\\
    While the length of $GP$ is less than $|P|$:\\+
        Let $p = \removemax(H)$.\\ 
        Let $p' = \farthest(p)$.\\
        Append $p'$ to $GP$.\\
        Insert $p'$ into the FVD.\\
        Update the keys of any cell whose radius changed.\\
        Insert $p$ and $p'$ into the heap.\\-
    Return $GP$.
\end{pseudo}

In line $6$, $\farthest(p)$ returns the farthest point in $V(p)$ to $p$.
So, for each iteration (lines $4-10$), the farthest point in the largest radius cell is appended.
The key to getting an improvement over the quadratic time algorithm is that in the greedy order, the degree of the neighbor graph is at most a constant.
Thus, only a constant number of cells are updated.
As Har-Peled and Mendel show~\cite{har-peled06fast}, the point location work can be amortized so that each point is touched at most $2^{O(d)}\log\spread$ times.
So, the total running time $2^{O(d)}n(\log n + \log \spread)$.
A more general version of this analysis is given in Section~\ref{sec:trees_to_permutations}.
One may observe that there are two parts of the algorithm that require superlinear work, the point location and the heap operations.
Our variations on Clarkson's Algorithm modify the implementation of these two aspects of the algorithm.

\subsection{Using a Bucket Queue}

We introduce a bucket queue~\cite{skiena98algorithms} to speed up heap operations in Clarkson's Algorithm.
It does not always return the true maximum radius cell.
However, we show that it still produces a correct greedy tree.
Thus, it will lead to the fast greedy tree construction in Section~\ref{sec:merge}.
The keys are rounded down to the nearest power of some user-specified constant $\bucketsize$.
Thus, instead of individual entries, the queue stores buckets of approximately equal keys.
When removing the max element, any entry in the largest bucket is returned.
\begin{quote}
    \textbf{The Heap Invariant:} 
    Let $q = \findmax(H)$, (ie. the site with the maximum key in $H$).
    For all $p\in H$, $\outradius(p) \le \greedy\dist(q, \farthest(q))$.
\end{quote}

Maintaining the Heap Invariant ensures that the sites are inserted in a $\vorapx\greedy$-greedy order.
Thus, the output is a greedy permutation.
If the FVD is only approximate, then the output will be an approximate greedy permutation.

We denote running the above algorithm on a set $P$ using $\lazy$-lazy point location, $\vorapx$-approximate cells, and $\greedy$-approximate heap as $\CLARKSON(P, \lazy, \vorapx, \greedy)$.
In Appendix~\ref{apx:fvd_properties} we prove structural properties of FVDs generated by this algorithm.
All the cells have bounded degree (Lemma~\ref{lem:degree_bound}), bounded aspect ratio (Lemma~\ref{lem:aspect_ratio_bound}) and the insertion order of sites is $\vorapx\greedy$-greedy (Lemma~\ref{lem:fvd_insertion_order}).
In Section~\ref{sec:bq} we modify the bucket queue to achieve constant time operations.

\subsection{Greedy Trees}\label{sec:greedy_trees}
To speed up point location, we modify the FVD to store clusters instead of individual points.
If all cells store only a constant number of clusters then each insertion takes constant time.
In this section we introduce the \emph{greedy tree}.
We use greedy tree nodes as clusters stored by the FVD.

A \emph{ball tree} on $P$ is a binary tree that represents a recursive partition of $P$.
The tree contains $2n-1$ nodes in total.
There is one leaf $\leaf{p}$ for each point $p$ of $P$.
Each node $z$ in $G$ represents a set $\points(p)$ of points corresponding to the leaves of its subtree.
The node $z$ has a center $\nodecenter{z}\in P$ and a radius $\noderad{z}$ such that
\[
    \points(z) \subseteq \ball(\nodecenter{z}, \noderad{z}),
\]
the ball centered at $\nodecenter{z}$ with radius $\noderad{z}$.
A ball tree is \emph{hereditary} if the center of each non-leaf node is the same as the center of its left child.
The two children partition the points contained in the parent.
A \emph{split} is a common operation on a ball tree where a node is replaced with its two children.
If $x$ and $y$ are the children, then their \emph{split distance} is $\splitdist(x) := \splitdist(y) = \dist(\nodecenter{x}, \nodecenter{y})$.
The greedy tree is constructed at the same time as the FVD when running Clarkson's algorithm.
Along the way, we store the \emph{predecessor mapping} on $P$, a function $\pred: P\setminus\{p_0\}\to P$.
Each time we insert a new site $p'\in \vor(p)$, we say that $p$ is the \emph{predecessor} of $p'$ and set $\pred(p') = p$.
The \emph{insertion distance} of a point $p$ is the distance $\e_p:=\dist(p, \pred(p))$.
To update the greedy tree after a split, we take $\leaf{p}$ and attach two new leaves, a left child centered at $p$ and a right child centered at $p'$.
We will not compute the exact radius of each node.
Theorem~\ref{thm:structure} below shows that the distance to the center of the right child gives an approximate radius.

An ordering of the points can be extracted from a tree starting with the center of the root node.
A max-heap of nodes is initialized with the root.
The heap key for a node $x$ with right child $y$ is $\dist(\nodecenter{x}, \nodecenter{y})$.
While the heap is nonempty, we remove the maximum node $x$ and append the center of its right child to the ordering.
The non-leaf children of $x$ are then added to the heap.
The resulting point order is called the \emph{heap-ordered traversal}.

Let $P$ be ordered with first point $p_0$.
For any point $a\in P$, let $P_a$ denote the prefix of $P$ up to but not including $a$.
Let $\pred$ be a predecessor mapping on $P$ such that $\pred(a) \in P_a$ for all $a\in P$.
An $(\scale, \gpapprox, \locallygreedy)$-\emph{greedy tree} $G$ is a hereditary ball tree with root centered at $p_0$, children defined by $\pred$, and has the following properties: 
\begin{enumerate}
    \item $\scale$-scaling: $\e_a \le \frac 1 \scale \e_{\pred(a)}$ for all $a \in P\setminus \{p_0\}$.
    \item $\gpapprox$-approximate:
        $\e_a \le \gpapprox\dist(a,P_a)$ for all $x\in P\setminus\{p_0\}$,
    \item $\locallygreedy$-greedy: A heap-ordered traversal of the tree gives a $\locallygreedy$-greedy permutation.
\end{enumerate}

\noindent For completeness, a structure theorem is included in Appendix~\ref{sec:apx_structure} (see also~\cite{chubet23proximity}).

\begin{figure}[ht]
\centering
\includegraphics[scale=0.45]{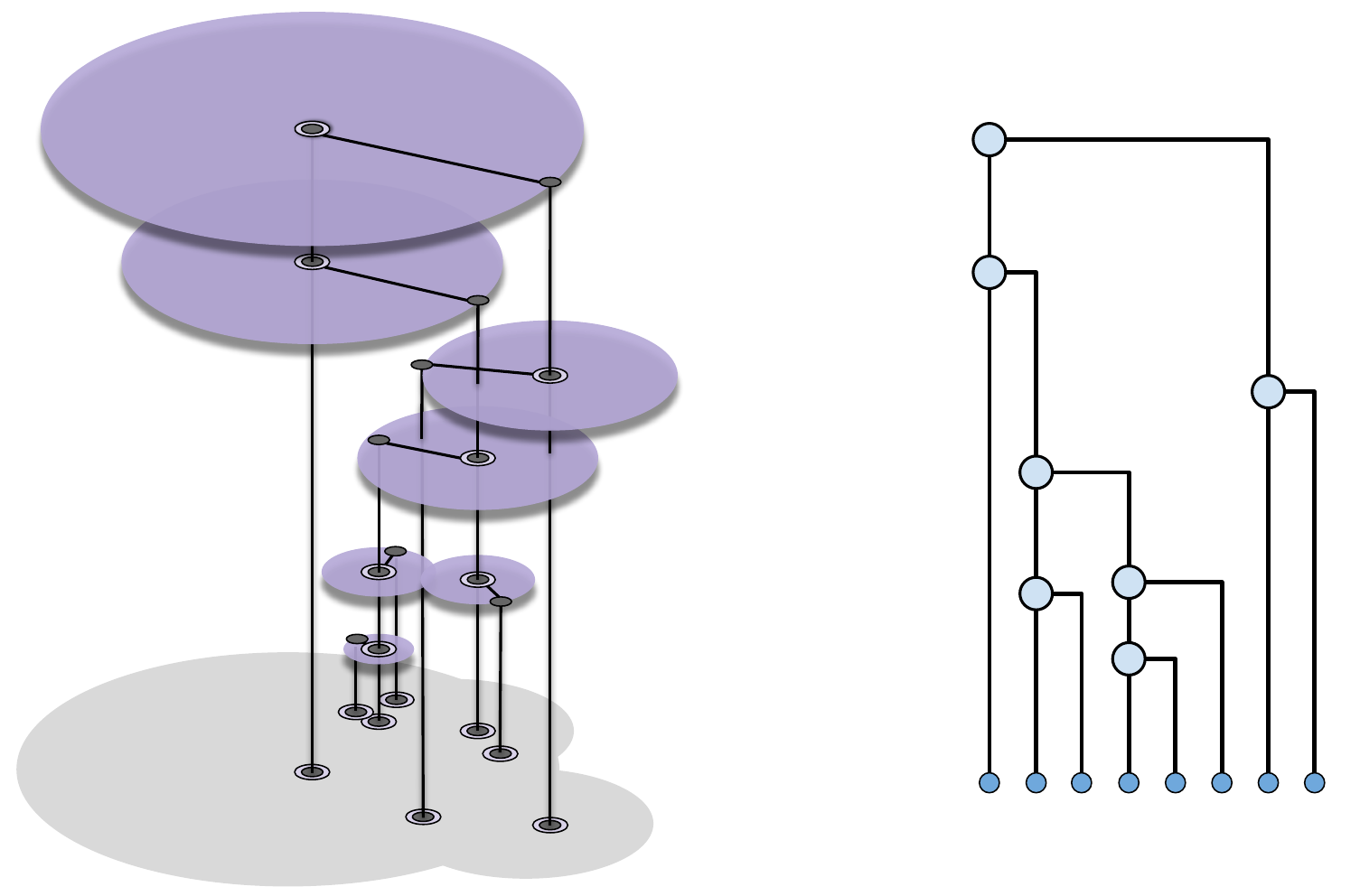}
\caption{ 
On the left is depicted a geometric representation of the greedy tree, with the corresponding data structure on the right.
}
\label{fig:mobile-greedy-tree}
\end{figure}

\subsection{Challenges to Linear-Time Tree Merge}\label{sec:main_results}

The preceding sections provide sufficient background to formally state the main goals of this paper and hint at the challenges.
The goals are the following.
\begin{itemize}
    \item Develop a simple variation of Clarkson's Algorithm to merge two greedy trees in $2^{O(d)}n$ time for doubling metrics.  This leads immediately to a $2^{O(d)}n\log n$ time algorithm to compute greedy trees by recursive merging.
    This also implies a $2^{O(d)}n$ parallel time algorithm.
    \item Use the greedy tree to construct $(1 + \frac{1}{n})$-approximate greedy permutations in $2^{O(d)}n\log n$ time.
\end{itemize}
The challenges to overcome are the following.
\begin{itemize}
    \item The cost of point location is too high.  In Clarkson's algorithm, the cost is $2^{O(d)}\log \spread$ per point.  Har-Peled and Mendel reduce this to $2^{O(d)}\log n$ per point at the cost of substantial complexity.  We will reduce it to $2^{O(d)}$ time per point when merging approximate greedy trees (see Section~\ref{sec:trees_to_permutations}). 
    \item The cost of maintaining a total order with a heap is too high to permit a linear time merge.  Some kind of relaxation of both the heap and (in the case of high spread) the ordering is necessary.  We will show that the greedy tree can be constructed without inserting nodes strictly in the greedy order.
\end{itemize}

    \section{From Greedy Trees to Greedy Permutations}\label{sec:trees_to_permutations}

In the simplest algorithms, the greedy tree is constructed while computing a greedy permutation.
However, if one is given a greedy tree first, a greedy permutation can be extracted by the heap-order traversal.
If the greedy tree was only an approximation, then the resulting greedy permutation is also only an approximation.
Specifically, the heap-order traversal of an $(\lazy, \gpapprox, \locallygreedy)$-greedy tree yields (by definition) a $\locallygreedy$-approximate greedy permutation.
However, it is possible to extract a very fine approximation to the greedy permutation from a very coarse approximate greedy tree.
We call this \emph{refining} a greedy tree.

The idea is to run Clarkson's algorithm storing greedy tree nodes rather than individual points in the Voronoi cells. 
This allows many points to be moved at once.
If a node is too large, it is simply replaced by its two children.
This is essentially the last step of Har-Peled and Mendel's algorithm~\cite{har-peled06fast} with the difference that a greedy tree is used rather than the hierarchical separator tree.
The advantage of using a greedy tree is that (as we show in Section~\ref{sec:results}) it is much simpler to compute.

We give the detailed refining algorithm and summarize its analysis in Section~\ref{sec:refining}.
We then explain how to amortize the point location cost in Section~\ref{sec:amortized_point_location}.
That section also reveals the origins of the $\log\spread$ term in the analysis of Clarkson's algorithm.

\subsection{Refining a Greedy Tree}\label{sec:refining}

In this section, we show how to run Clarkson's algorithm when the input has already been processed into an approximate greedy tree.
The algorithm is a slight modification to an algorithm of Har-Peled and Mendel modified to use a greedy tree rather than a hierarchical separator tree.
We include it here both for completeness and also because it introduces a key idea for the main algorithm of the next section.
Specifically, it shows that a greedy tree can speed up point location by moving nodes rather than points.
The refining algorithm below shows that to achieve $2^{O(d)}n\log n$ running time for $(1 + \frac{1}{n})$-approximate greedy permutations, it suffices to compute a relatively coarse approximate greedy tree.

The refining algorithm works as follows.
The input is an $(\scale, \gpapprox, \locallygreedy)$-greedy tree $G$.

\begin{pseudo}*
    \hd{\REFINE}(G): \\
    Initialize an FVD with the root point of $G$.\\
    Initialize a heap $A$ to store cells by out-radius as in Clarkson's Algorithm.\\
    Initialize a heap $B$ to perform a heap-ordered traversal of $G$.\\
    While $A$ is non-empty:\\+
        While $\maxkey(B)\le \frac{1}{n}\maxkey(A)$:\\+
        Split the max of $B$ and reinsert the children.\\ 
        Locate the new points in the FVD.\\-
        Insert $\farthest(p)$ as a new site in the FVD, where $p = \findmax(A)$.\\
        Append $\farthest(p)$ to the output permutation.
\end{pseudo}

The \REFINE algorithm computes a $(1+\frac{1}{n})$-approximate greedy permutation in $2^{O(d)}n\log n$ time.
Splitting nodes of $B$ reveals points as the algorithm proceeds.
At each step, we insert the point farthest from its site among those points revealed by the traversal of $G$.
If the last insertion was at distance $r$, then any points not yet revealed are within $\frac{r}{n}$ of one that has.
This guarantees that the output order is $(1+\frac{1}{n})$-approximate greedy.
We analyze the running time in the next section.

\subsection{Amortized Point Location Analysis}\label{sec:amortized_point_location}

In this section, we analyze the point location cost of the \REFINE algorithm.
In summary, Lemma~\ref{lem:amortized_point_location} is an amortization of the touches per node during point location.
Specifically, within an annulus of radii $r$ to $2r$, a given node incurs $2^{O(d)}$ touches.
\begin{restatable}{lem}{amortizepl}\label{lem:amortized_point_location}
    For $r > 0$, let $X_r$ denote the set of sites $p$ that touch $x$ such that $r\le \dist(\nodecenter{x}, p) < 2r$.
Then $|X_r|$ is $2^{O(d)}$.
\end{restatable}
\begin{figure}
    \centering
    \includegraphics[scale=0.6]{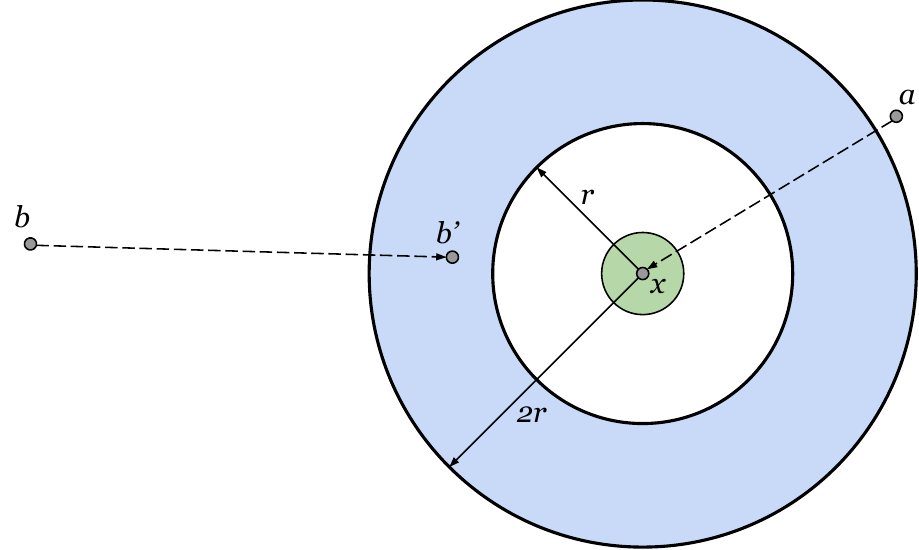}
    \caption{A node $x\in V(a)$ is touched when site $b'$ is inserted.
        The predecessor of $b'$ is $b$.
        We bound the packing of the set $X_r$ by finding a lower bound for $\dist(b',b)$.
    }
\end{figure}
A detailed proof of Lemma~\ref{lem:amortized_point_location} can be found in Appendix~\ref{apx:amortize_pl}.
Then, we argue that a node appears only when the insertion distance is already within a factor of $n$ of its radius.
As long as the greedy tree satisfies the strong packing property (see Appendix~\ref{sec:strong_packing}), a node can only be touched $O(\log n)$ times before it is inserted or split again.
Removing and splitting the max of $B$ takes $O(n\log n)$ time due to the heap operations on $B$.
It also requires at most one heap operation on $A$ for each point (to update the radius).

\begin{restatable}{restate}{refineanalysis}\label{thm:refine_analysis}
    Let $G$ be a greedy tree.
    Then $\REFINE(G)$ produces a $(1+\frac{1}{n})$-greedy permutation in $2^{O(d)}n\log n$ time.
\end{restatable}

\begin{proof}
    Consider a step in the algorithm where a point $a$ is added to the output order.
    Let $P_a$ be sites inserted before $a$.
    We want to show that $a$ was the next point in an approximate greedy permutation.
    Let $P'$ be the points revealed so far.
    Let $p$ be any point.
    Let $p'$ be the nearest point in $P'$ to $p$.
    Then, because the heap $A$ at this time is ordered by distance to $P_a$,
    \begin{align*}
        \dist(p,P_a)
            &\le \dist(p, p') + \dist(p', P_a)\\
            &\le \maxkey(B) + \dist(a, P_a)\\
            &\le \frac{1}{n} \maxkey(A) + \dist(a, P_a)\\
            &\le \left(1+\frac{1}{n}\right) \dist(a, P_a)
    \end{align*}
    Thus, the output ordering is $(1+\frac{1}{n})$-greedy.

    For the running time, we first observe that the loop that traverses $B$ takes $2^{O(d)} + O(\log n)$ time per point.
    The $2^{O(d)}$ term comes from the point location when iterating over the neighbors of the site of the predecessor.
    The $O(\log n)$ term comes from the heap operations.
    The traversal of $B$ includes one removal and at most two insertions.
    Also, there is one key update in the $A$ heap because the newly revealed point might change the radius of its site.

    The main iteration of the algorithm is an incremental FVD construction for points added in an approximate greedy order.
    By Lemma~\ref{lem:degree_bound}, the neighbor graph has degree $2^{O(d)}$ in such cases.
    For each point, this requires $O(\log n)$ time to find the max in the heap, $2^{O(d)}$ time to iterate over the neighbors to update and prune the neighbor graph.
    
    The only remaining cost is the point location.
    We amortize the point location cost using Lemma~\ref{lem:amortized_point_location} by showing that each node $x$ is touched $2^{O(d)}$ times from each annulus.
    In line $5$, the \REFINE algorithm guarantees if the split distance is $r$, then the farthest touch must come from cells with outradii at most $nr$.
    Let $\pi$ be the strong packing constant for $G$.
    If there is a touch within distance $\pi r$, then by Theorem~\ref{thm:structure}, that point must be a descendant of $x$, which means $x$ must have been split.
    Thus, there are at most $\log(\frac{nr}{\pi r}) = O(\log n)$ such annuli.
\end{proof}

    \section{The Main Results}\label{sec:results}

Using the results of the previous two sections, we are now ready to prove the main results of this paper.
The first step is to use a greedy tree to speed up the point location step (Section~\ref{sec:local_point_location}).
The second step is to modify the heap that orders the insertions in Clarkson's Algorithm (Section~\ref{sec:bq}).
Finally, we show that running Clarkson's algorithm where the input is given as a pair of greedy trees and the new heap structure will merge them into a greedy tree in linear time (Section~\ref{sec:merge}).

\subsection{Point Location with a Greedy Tree}\label{sec:local_point_location}

While running Clarkson's Algorithm we store greedy tree nodes in the cells.
The point location is then performed by moving nodes rather than individual points.
The pruning condition for edges in the neighbor graph uses the out-radii of the cells.
The out-radii are also used as the keys in the max heap.
If a node $x$ is in a cell, then every point of $\points(x)$ is in the cell.
Because $\points(x)$ are abstracted by $x$, the out-radii of the cells are only approximate.
If a node is too big to give a good approximation to the radius, then it is split.
We call this \emph{tidying} a cell.
During point location, if the radius of a node is too big to be placed unambiguously in one cell or another, it is split.
We call this \emph{split-on-move}.

\paragraph*{Estimating the Out-Radius}

Since cells store greedy tree nodes, we can only approximate their out-radii.
For a cell $\vor(p)$, we keep the greedy tree nodes in a max-heap $H_p$ with keys $\key_p(x) := \dist(p,\nodecenter{x})+\noderad{x}$.
Then $\key_p(x)$ is an upper bound on the distance of any $x'\in \points(x)$ to $p$.
It follows that $\maxkey_p$ is an upper bound for $\outradius(p)$, the out-radius of $p$.
To maintain the heap invariant, we keep all cells sufficiently tidy (see Figure~\ref{fig:node_PL}).
\begin{quote}
    \textbf{Tidying $\vor(p)$:}
    While $\maxkey_p > c \dist(p, \farthest(p))$, remove the maximum node of $H_p$, split it and update $H_p$.
\end{quote}
Determining the constant for tidying is a minute implementation detail, but for simplicity we set it equal to the heap bucket size.
\begin{restatable}{lem}{tidycorrect}[Tidying Lemma]\label{lem:tidying_correctness}
    Given a FVD constructed with $\radapprox$-tidy cells, and a heap with $\bucketsize$-approximate buckets, the Heap Invariant holds for $\greedy \ge \radapprox^2$.
\end{restatable}
\begin{proof}
    Let $p = \max(H)$.
    For any site $q\in H$,
    \begin{align*}
        \outradius(q) & \le \key(q) \because{by definition of $\key(q)$}\\
        & \le \bucketsize\key(p) \because{$\bucketsize$-approximate heap}\\
        & \le \bucketsize^2\dist(p,\farthest(p)) \because{$\bucketsize$-tidy cells}& \qedhere
    \end{align*}
\end{proof}


\begin{figure}
\centering
\begin{subfigure}{0.45\textwidth}
\begin{tikzpicture}[scale=0.08]
\draw (-21,15) rectangle (53,-59);

\draw[fill={rgb,1:red,0.5;green,0.8;blue,0.5}, stroke={rgb,1:red,0.6;green,0.6;blue,1}, line width=0.5, opacity=0.5] (15.64, -21.74) circle (35);

\draw[fill={rgb,1:red,1;green,0.3;blue,0.3}, stroke={rgb,1:red,0;green,0;blue,0}, line width=0.5] (15.64, -21.74) circle (0.5);
\end{tikzpicture}
\subcaption{A Voronoi cell containing one greedy tree node.}
\end{subfigure}
\begin{subfigure}{0.45\textwidth}
\begin{tikzpicture}[scale=0.08]

\draw (-21,15) rectangle (53,-59);

\draw[fill={rgb,1:red,0.5;green,0.8;blue,0.5}, stroke={rgb,1:red,0.6;green,0.6;blue,1}, line width=0.5, opacity=0.5] (39.04, -40.80) circle (1.06);
\draw[fill={rgb,1:red,0.5;green,0.8;blue,0.5}, stroke={rgb,1:red,0.6;green,0.6;blue,1}, line width=0.5, opacity=0.5] (37.14, -41.28) circle (1.07);
\draw[fill={rgb,1:red,0.5;green,0.8;blue,0.5}, stroke={rgb,1:red,0.6;green,0.6;blue,1}, line width=0.5, opacity=0.5] (34.70, -41.30) circle (1.45);
\draw[fill={rgb,1:red,0.5;green,0.8;blue,0.5}, stroke={rgb,1:red,0.6;green,0.6;blue,1}, line width=0.5, opacity=0.5] (35.74, -42.70) circle (0.51);
\draw[fill={rgb,1:red,0.5;green,0.8;blue,0.5}, stroke={rgb,1:red,0.6;green,0.6;blue,1}, line width=0.5, opacity=0.5] (36.82, -38.72) circle (0.91);
\draw[fill={rgb,1:red,0.5;green,0.8;blue,0.5}, stroke={rgb,1:red,0.6;green,0.6;blue,1}, line width=0.5, opacity=0.5] (41.16, -33.42) circle (3.45);
\draw[fill={rgb,1:red,0.5;green,0.8;blue,0.5}, stroke={rgb,1:red,0.6;green,0.6;blue,1}, line width=0.5, opacity=0.5] (38.48, -43.54) circle (2.78);
\draw[fill={rgb,1:red,0.5;green,0.8;blue,0.5}, stroke={rgb,1:red,0.6;green,0.6;blue,1}, line width=0.5, opacity=0.5] (39.38, -38.76) circle (1.06);
\draw[fill={rgb,1:red,0.5;green,0.8;blue,0.5}, stroke={rgb,1:red,0.6;green,0.6;blue,1}, line width=0.5, opacity=0.5] (33.84, -35.62) circle (2.93);
\draw[fill={rgb,1:red,0.5;green,0.8;blue,0.5}, stroke={rgb,1:red,0.6;green,0.6;blue,1}, line width=0.5, opacity=0.5] (41.06, -40.22) circle (1.75);
\draw[fill={rgb,1:red,0.5;green,0.8;blue,0.5}, stroke={rgb,1:red,0.6;green,0.6;blue,1}, line width=0.5, opacity=0.5] (34.04, -38.64) circle (2.46);
\draw[fill={rgb,1:red,0.5;green,0.8;blue,0.5}, stroke={rgb,1:red,0.6;green,0.6;blue,1}, line width=0.5, opacity=0.5] (15.64, -21.74) circle (27.97);
\draw[fill={rgb,1:red,0.5;green,0.8;blue,0.5}, stroke={rgb,1:red,0.6;green,0.6;blue,1}, line width=0.5, opacity=0.5] (37.92, -36.72) circle (2.46);

\draw[fill=none, stroke={rgb,1:red,0.5;green,0.3;blue,0.6}, line width=1](15.64, -21.74) -- (41.16, -33.42) -- (41.06, -40.22) -- (38.48, -43.54) -- (35.74, -42.70) -- (15.64, -21.74) -- (15.64, -21.74);

\draw[fill={rgb,1:red,1;green,1;blue,1}, stroke={rgb,1:red,0;green,0;blue,0}, line width=0.5, opacity=1] (34.04, -38.64) circle (0.3);
\draw[fill={rgb,1:red,1;green,1;blue,1}, stroke={rgb,1:red,0;green,0;blue,0}, line width=0.5, opacity=1] (41.06, -40.22) circle (0.3);
\draw[fill={rgb,1:red,1;green,1;blue,1}, stroke={rgb,1:red,0;green,0;blue,0}, line width=0.5, opacity=1] (33.84, -35.62) circle (0.3);
\draw[fill={rgb,1:red,1;green,1;blue,1}, stroke={rgb,1:red,0;green,0;blue,0}, line width=0.5, opacity=1] (39.38, -38.76) circle (0.3);
\draw[fill={rgb,1:red,1;green,1;blue,1}, stroke={rgb,1:red,0;green,0;blue,0}, line width=0.5, opacity=1] (15.64, -21.74) circle (0.3);
\draw[fill={rgb,1:red,1;green,1;blue,1}, stroke={rgb,1:red,0;green,0;blue,0}, line width=0.5, opacity=1] (37.92, -36.72) circle (0.3);
\draw[fill={rgb,1:red,1;green,1;blue,1}, stroke={rgb,1:red,0;green,0;blue,0}, line width=0.5, opacity=1] (38.48, -43.54) circle (0.3);
\draw[fill={rgb,1:red,1;green,1;blue,1}, stroke={rgb,1:red,0;green,0;blue,0}, line width=0.5, opacity=1] (41.16, -33.42) circle (0.3);
\draw[fill={rgb,1:red,1;green,1;blue,1}, stroke={rgb,1:red,0;green,0;blue,0}, line width=0.5, opacity=1] (36.82, -38.72) circle (0.3);
\draw[fill={rgb,1:red,1;green,1;blue,1}, stroke={rgb,1:red,0;green,0;blue,0}, line width=0.5, opacity=1] (35.74, -42.70) circle (0.3);
\draw[fill={rgb,1:red,1;green,1;blue,1}, stroke={rgb,1:red,0;green,0;blue,0}, line width=0.5, opacity=1] (34.70, -41.30) circle (0.3);
\draw[fill={rgb,1:red,1;green,1;blue,1}, stroke={rgb,1:red,0;green,0;blue,0}, line width=0.5, opacity=1] (37.14, -41.28) circle (0.3);
\draw[fill={rgb,1:red,1;green,1;blue,1}, stroke={rgb,1:red,0;green,0;blue,0}, line width=0.5, opacity=1] (39.04, -40.80) circle (0.3);

\draw[fill={rgb,1:red,1;green,0.3;blue,0.3}, stroke={rgb,1:red,0;green,0;blue,0}, line width=0.5] (15.64, -21.74) circle (0.5);
\end{tikzpicture}
\subcaption{A cell is tidied; nodes are split until an approximate farthest is determined.}
\end{subfigure}

\vspace{2em}

\begin{subfigure}{0.45\textwidth}
\begin{tikzpicture}[scale=0.08]

\draw (-21,15) rectangle (53,-59);

\draw[fill={rgb,1:red,0.5;green,0.8;blue,0.5}, stroke={rgb,1:red,0.6;green,0.6;blue,1}, line width=0.5, opacity=0.5] (39.04, -40.80) circle (1.06);
\draw[fill={rgb,1:red,0.5;green,0.8;blue,0.5}, stroke={rgb,1:red,0.6;green,0.6;blue,1}, line width=0.5, opacity=0.5] (37.14, -41.28) circle (1.07);
\draw[fill={rgb,1:red,0.5;green,0.8;blue,0.5}, stroke={rgb,1:red,0.6;green,0.6;blue,1}, line width=0.5, opacity=0.5] (34.70, -41.30) circle (1.45);
\draw[fill={rgb,1:red,0.5;green,0.8;blue,0.5}, stroke={rgb,1:red,0.6;green,0.6;blue,1}, line width=0.5, opacity=0.5] (35.74, -42.70) circle (0.51);
\draw[fill={rgb,1:red,0.5;green,0.8;blue,0.5}, stroke={rgb,1:red,0.6;green,0.6;blue,1}, line width=0.5, opacity=0.5] (36.82, -38.72) circle (0.91);
\draw[fill={rgb,1:red,0.5;green,0.8;blue,0.5}, stroke={rgb,1:red,0.6;green,0.6;blue,1}, line width=0.5, opacity=0.5] (41.16, -33.42) circle (3.45);
\draw[fill={rgb,1:red,0.5;green,0.8;blue,0.5}, stroke={rgb,1:red,0.6;green,0.6;blue,1}, line width=0.5, opacity=0.5] (38.48, -43.54) circle (2.78);
\draw[fill={rgb,1:red,0.5;green,0.8;blue,0.5}, stroke={rgb,1:red,0.6;green,0.6;blue,1}, line width=0.5, opacity=0.5] (39.38, -38.76) circle (1.06);
\draw[fill={rgb,1:red,0.5;green,0.8;blue,0.5}, stroke={rgb,1:red,0.6;green,0.6;blue,1}, line width=0.5, opacity=0.5] (33.84, -35.62) circle (2.93);
\draw[fill={rgb,1:red,0.5;green,0.8;blue,0.5}, stroke={rgb,1:red,0.6;green,0.6;blue,1}, line width=0.5, opacity=0.5] (41.06, -40.22) circle (1.75);
\draw[fill={rgb,1:red,0.5;green,0.8;blue,0.5}, stroke={rgb,1:red,0.6;green,0.6;blue,1}, line width=0.5, opacity=0.5] (34.04, -38.64) circle (2.46);
\draw[fill={rgb,1:red,0.5;green,0.8;blue,0.5}, stroke={rgb,1:red,0.6;green,0.6;blue,1}, line width=0.5, opacity=0.5] (15.64, -21.74) circle (27.97);
\draw[fill={rgb,1:red,0.5;green,0.8;blue,0.5}, stroke={rgb,1:red,0.6;green,0.6;blue,1}, line width=0.5, opacity=0.5] (37.92, -36.72) circle (2.46);

\draw[fill=none, stroke={rgb,1:red,0.5;green,0.3;blue,0.6}, line width=1](33.84, -35.62) -- (41.16, -33.42) -- (41.06, -40.22) -- (38.48, -43.54) -- (35.74, -42.70) --  (34.70, -41.30) -- (34.04, -38.64) -- (33.84, -35.62) -- cycle;

\draw[fill={rgb,1:red,1;green,1;blue,1}, stroke={rgb,1:red,0;green,0;blue,0}, line width=0.5, opacity=1] (34.04, -38.64) circle (0.3);
\draw[fill={rgb,1:red,1;green,1;blue,1}, stroke={rgb,1:red,0;green,0;blue,0}, line width=0.5, opacity=1] (41.06, -40.22) circle (0.3);
\draw[fill={rgb,1:red,1;green,1;blue,1}, stroke={rgb,1:red,0;green,0;blue,0}, line width=0.5, opacity=1] (33.84, -35.62) circle (0.3);
\draw[fill={rgb,1:red,1;green,1;blue,1}, stroke={rgb,1:red,0;green,0;blue,0}, line width=0.5, opacity=1] (39.38, -38.76) circle (0.3);
\draw[fill={rgb,1:red,1;green,1;blue,1}, stroke={rgb,1:red,0;green,0;blue,0}, line width=0.5, opacity=1] (15.64, -21.74) circle (0.3);
\draw[fill={rgb,1:red,1;green,1;blue,1}, stroke={rgb,1:red,0;green,0;blue,0}, line width=0.5, opacity=1] (37.92, -36.72) circle (0.3);
\draw[fill={rgb,1:red,1;green,1;blue,1}, stroke={rgb,1:red,0;green,0;blue,0}, line width=0.5, opacity=1] (38.48, -43.54) circle (0.3);
\draw[fill={rgb,1:red,1;green,1;blue,1}, stroke={rgb,1:red,0;green,0;blue,0}, line width=0.5, opacity=1] (41.16, -33.42) circle (0.3);
\draw[fill={rgb,1:red,1;green,1;blue,1}, stroke={rgb,1:red,0;green,0;blue,0}, line width=0.5, opacity=1] (36.82, -38.72) circle (0.3);
\draw[fill={rgb,1:red,1;green,1;blue,1}, stroke={rgb,1:red,0;green,0;blue,0}, line width=0.5, opacity=1] (35.74, -42.70) circle (0.3);
\draw[fill={rgb,1:red,1;green,1;blue,1}, stroke={rgb,1:red,0;green,0;blue,0}, line width=0.5, opacity=1] (34.70, -41.30) circle (0.3);
\draw[fill={rgb,1:red,1;green,1;blue,1}, stroke={rgb,1:red,0;green,0;blue,0}, line width=0.5, opacity=1] (37.14, -41.28) circle (0.3);
\draw[fill={rgb,1:red,1;green,1;blue,1}, stroke={rgb,1:red,0;green,0;blue,0}, line width=0.5, opacity=1] (39.04, -40.80) circle (0.3);

\draw[fill={rgb,1:red,1;green,0.3;blue,0.3}, stroke={rgb,1:red,0;green,0;blue,0}, line width=0.5] (15.64, -21.74) circle (0.5);
\end{tikzpicture}
\subcaption{Some nodes can be located, but others are two large to do so unambiguously.}
\end{subfigure}
\begin{subfigure}{0.45\textwidth}
\input{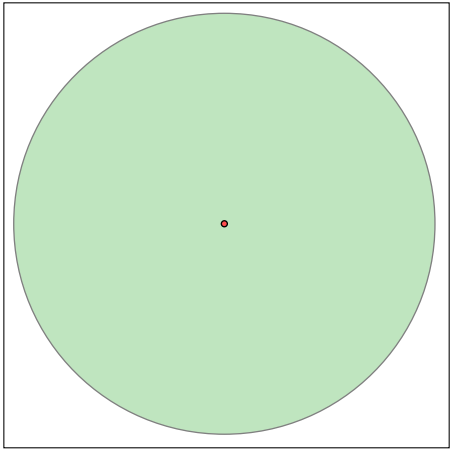}
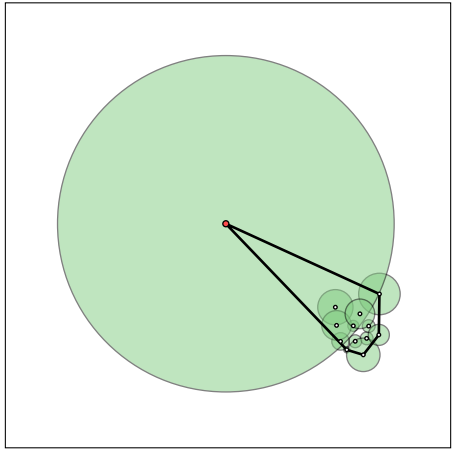
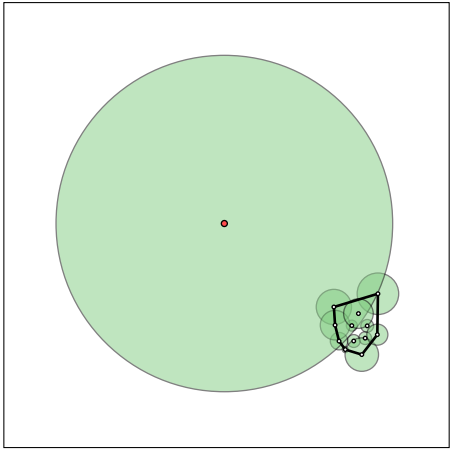
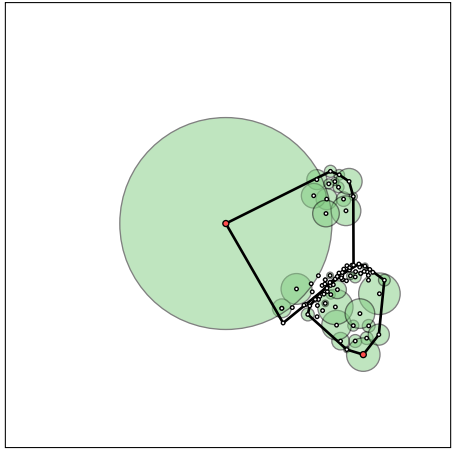
\subcaption{Some nodes are split-on-move; nodes are split until they can be unamiguously located.}
\end{subfigure}
    \caption{
        Point location of greedy tree nodes.
    }
    \label{fig:node_PL}
\end{figure}

\paragraph*{Moving a Node to a New Cell }

When inserting a new site $p'$ into a $\vorapx$-approximate FVD with $\lazy$-lazy moves, an uninserted point $q\in \vor(p)$ must move if $\vorapx\dist(q,p') < \dist(q,p)$ and must not move if $\lazy\dist(q,p') \ge \dist(q,p)$.
To satisfy both conditions when a whole node of points moves together, the nodes have to be small enough so the choice to move or not is unambiguous.
This operation is depicted in Figure~\ref{fig:node_PL}.

\begin{quote}
    \textbf{Split-on-Move:}
    Let $p'$ be a newly inserted site in an FVD.
    Let $x$ be a node in $\vor(p)$ with center $a$ and radius upper bounded by $r$.
    \begin{itemize}
        \item \textbf{Move.} First, we check if $x$ can move to $\vor(p)$.
        If $\lazy(\dist(a,p') + r) \le \dist(a,p) - r$, then every point of $\points(x)$ can move.
    \item \textbf{Stay.} Otherwise, we check if $x$ can stay in $\vor(p)$.
        If $\vorapx(\dist(a,p') - r) \ge \dist(a,p) + r$, then every point of $\points(x)$ can stay.
    \item \textbf{Split.} Else, $\points(x)$ may need to be partitioned between $\vor(p)$ and $\vor(p')$.
        In this case, we split $x$ and locate the children recursively.
    \end{itemize}
\end{quote}

A consequence of the split-on-move and tidying procedures is that the number of nodes in every cell is bounded by a constant.
A proof of the following lemma and a complete analysis of cell complexity when running $\CLARKSON(G, \lazy, \vorapx, \greedy)$ is presented in Section~\ref{sec:tidying}.
\begin{restatable}{lem}{constcellcomplexity}[Sparsity Lemma]\label{lem:constant_complexity}
    Let $\vor(p)$ be a cell in an FVD when running $\CLARKSON(G, \lazy, \vorapx, \greedy)$.
    The number of nodes in $\vor(p)$ is $\zeta^{O(d)}$ for a constant $\zeta$ (that depends on $\lazy$, $\vorapx$, and $\greedy$).
\end{restatable}

\subsection{Eliminating the Dependence on the Spread}\label{sec:bq}

In the case of super-polynomial spread, there could still be $O(\log_{\bucketsize} \spread)$ nonempty buckets.
To eliminate the dependence on the spread, we add an auxiliary list called the \emph{backburner}.
Any time an entry is added to the bucket queue with a key that is more than some constant times smaller than the largest key, it is appended to the backburner.
If the main queue is empty, then the bucket queue will return an arbitrary element from the backburner.
If the maximum keys are monotonely non-increasing there will never be more than a constant number of buckets in the queue.
Moreover, any time a site is placed on the backburner, there is guaranteed to be a large empty annulus separating that cell from the rest of the points (Lemma~\ref{lem:bucket_queue}).
Along with the pruning condition, the empty annulus guarantees that if a site is processed from the backburner, it will be isolated in the neighbor graph (Lemma~\ref{lem:bb_isolation}).

A difficulty does arise with the backburner.
If a point is inserted from a site that was on the backburner, then it may not truly be the next point in the greedy order.
However, Lemma~\ref{lem:bb_isolation} guarantees that the site has no neighbors in the neighbor graph.
In the FVD construction, the only way for an insertion to affect another cell is if there is a path between them in the neighbor graph.
Disconnected components of the neighbor graph can be processed in any order, or even in parallel.
So, although the order is no longer greedy, the output tree \emph{is} greedy.
(See Theorem~\ref{thm:bb_same_tree}.)

\subsection{Backburner Analysis}
Let $\CLARKSONBB(G, \lazy, \vorapx, \greedy)$ refer to running $\CLARKSON(G, \lazy, \vorapx, \greedy)$ where the main heap is a $\bucketsize = \sqrt{\greedy}$-bucket queue with a backburner.
The following lemma shows that if a cell is placed onto the backburner, there must be an empty annulus around that cell containing no points from any other cell.
See Appendix~\ref{sec:apx_bb} for full proofs of this section.

\begin{restatable}{lem}{bucketqueue}[Empty Annulus Lemma]\label{lem:bucket_queue}
  Let $\vor(p)$ be a cell in an FVD when running $\CLARKSONBB(G, \lazy, \vorapx, \greedy)$ with a bucket queue with $\bucketsize$-approximate buckets and of constant length $s$.
  If $\vor(p)$ goes on the backburner, then there are no points $a \in P$ such that,
  \[
    \outradius(p) < \dist(p,a) \le \frac{\bucketsize^{s-1}}{\vorapx(1+\vorapx)} \outradius(p).
    \]
\end{restatable}

\begin{figure}[ht]
    \centering
    \includegraphics[scale=0.5]{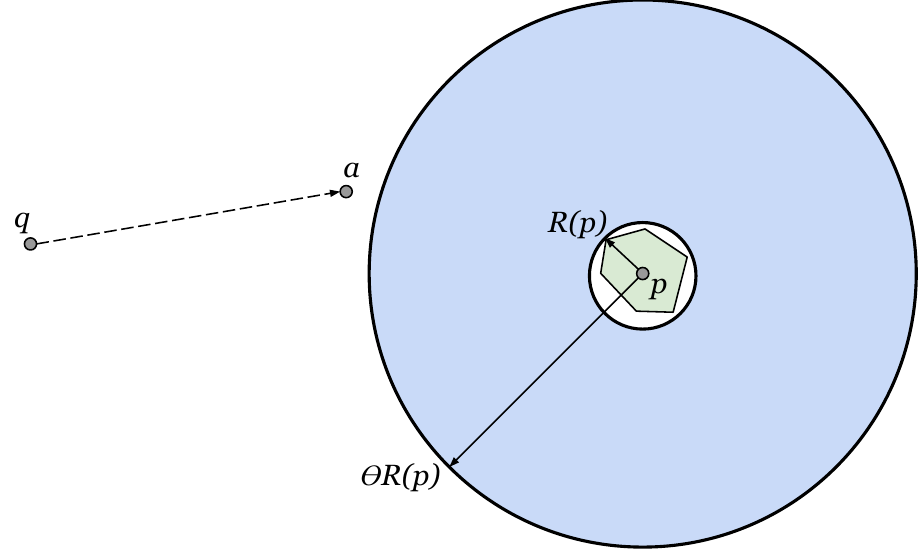}
    \caption{
        No point $a\not\in V(q)$ can exist within a distance $\annulus\outradius(p)$ of a site $p$ on the backburner.
        So for any site on the backburner, there exists an empty annulus where no points reside.
    }
    \label{fig:emptyannulus}
\end{figure}

We refer to $\annulus = \frac{\bucketsize^{s-1}}{\vorapx(1+\vorapx)}$ as the empty annulus constant.
Next, we show that when a cell is removed from the backburner, it has no neighbors in the neighbor graph.

\begin{restatable}{lem}{bbisolation}\label{lem:bb_isolation}[Isolation Lemma]
    When running $\CLARKSONBB(G, \lazy, \vorapx, \greedy)$ with $\annulus > 3$, if a cell $\vor(p)$ is removed from the backburner then it is isolated.
\end{restatable}

Next, we formally prove that the tree constructed with a bucket queue is still greedy regardless of whether or not a backburner was used.
We will use the following lemma.

\begin{restatable}{lem}{covering}\label{lem:covering}
    Let $G$ be the greedy tree constructed by $\CLARKSONBB(P, \lazy, \vorapx, \greedy)$.
    Let $P_a$ denote the partial heap-order traversal of $G$ just before $a$ is to be appended to it.
    For all $p \in P$ we have $\dist(p, P_a) \le \greedy^2 \e_a$.
\end{restatable}

Now, we show that the order in which the cells are removed from the backburner does not affect the constructed greedy tree.

\begin{theorem}\label{thm:bb_same_tree}
    The greedy tree $G$ constructed by $\CLARKSONBB(P, \lazy, \vorapx, \greedy)$ is an $(\frac{\lazy}{\greedy}, \vorapx, \vorapx\greedy^2)$-greedy tree.
\end{theorem}
\begin{proof}
    By Corollary~\ref{lem:cell_scaling}, insertion distances are $\frac{\lazy}{\greedy}$-scaling.
    It needs to be shown that the heap-order traversal of $G$ is $\vorapx\greedy^2$-greedy.
    
    Let $P_a$ denote the partial heap-order traversal of $G$ just before $a$ is to be appended to it.
    It is sufficient to show that $\dist(p, P_a) \le \vorapx\greedy^2 \dist(a, P_a)$ for an arbitrary $p \not \in P_a$.
    Let $a' := \arg\min_{q \in P_a} \{\e_{q}\}$ and let $\emin$ denote this insertion distance.
    As $a'$ precedes $a$, we have, $\dist(p ,P_a) \le \dist(p, P_{a'}) \le \greedy^2 \emin$.
    The second inequality follows from Lemma~\ref{lem:covering}.

    Let $p'$ be the nearest neighbor of $a$ in $P_a$.
    Then, by the Packing property from Theorem~\ref{thm:structure} we have, $\dist(a, P_a) = \dist(p', a) \ge \frac{1}{\vorapx} \min\{\e_a, \e_{p'}\} \ge \frac{1}{\vorapx} \emin$.

    It follows that, $\dist(p, P_a) \le \greedy^2 \emin \le \vorapx\greedy^2 \dist(a,P_a)$.
    Therefore, the heap-order traversal of $G$ produces a $\vorapx\greedy^2$-approximate greedy permutation.
\end{proof}

\subsection{Merging Greedy Trees in Linear Time}\label{sec:merge}

Let $A$ and $B$ be two $(\scale, \gpapprox, \locallygreedy)$-greedy trees with $\scale < \gpapprox$.
The \MERGE algorithm merges $A$ and $B$ into a single greedy tree with the same constants by running Clarkson's algorithm with two changes.
Point location is done using the nodes of $A$ and $B$, keeping the cells tidy as in Section~\ref{sec:local_point_location}.
The max heap is replaced with a bucket queue with a backburner as in Section~\ref{sec:bq}.

The algorithm runs in linear time.
At each step, choosing the next point to insert from the bucket queue takes constant time.
By Lemma~\ref{lem:constant_complexity}, the tidy cells contain only a constant number of nodes from each tree.
By Lemma~\ref{lem:degree_bound} the degree of each node is constant, so only a constant number of nodes are touched for each insertion.
So, each insertion takes constant time. 
Computing exact node radii might still take $2^{O(d)}n\log \spread$ time.
However, by Theorem~\ref{thm:structure} the distance to the right child can be used to compute an upper bound on the node radius.
Therefore, all radii can be approximated in linear time.
This gives the following theorem.

\begin{theorem}\label{thm:merge}
    Two $(\scale, \gpapprox, \greedy)$-greedy trees on $n$ points can be merged in $2^{O(d)}n$ time.
\end{theorem}

Finally, we show how \REFINE and \MERGE can be used to compute a $(1 + \frac{1}{n})$-greedy permutation in $2^{O(d)}n\log n$ time.

\begin{theorem}
    Given a metric space $P$ with $n$ points and doubling dimension $d$,
\begin{enumerate}
 \item an $(\scale, \gpapprox, \locallygreedy)$-greedy tree on $P$ can be constructed in $2^{O(d)}n\log n$ time.
    \item a $(1 + \frac{1}{n})$-greedy permutation of $P$ can be computed in $2^{O(d)}n\log n$ time.
    \item a $(\scale,\vorapx,\greedy)$-greedy tree can be computed in $2^{O(d)}n$ parallel time.
\end{enumerate}
\end{theorem}

\begin{proof}
    Split the input into two sets of size $\frac{n}{2}$ and recursively construct a tree on each.
    Then, run \MERGE to merge these into the desired tree.
    Theorem~\ref{thm:merge} says that the merge step takes $2^{O(d)}n$ time.
    So, the recursive construction takes $2^{O(d)}n\log n$ time in total, or $2^{O(d)}n$ time in parallel.

    To construct a $(1 + \frac{1}{n})$-greedy permutation, one first constructs the greedy tree $G$ with $\scale$ sufficiently large to have a non-trivial strong packing constant $\pi$.
    The desired permutation is $\REFINE(G)$.
    Both steps take $2^{O(d)}n\log n$ time, so that is the final running time.
\end{proof}

    \section{Conclusion}

In this paper, we showed that a simple modification of Clarkson's Algorithm for greedy permutations can be used to merge greedy trees in linear time.
The modification uses a simple bucket queue with a backburner that supports constant-time operations.
The greedy trees are used to speed up the point location.
The main insight is that the greedy tree does not need to be constructed in the greedy order.
The result is a simple merging algorithm that computes greedy trees and greedy permutations in $2^{O(d)}n\log n$ time.

Within computational geometry, there are a variety of results that allow one to reduce a $\log \spread$ term to $\log n$.
For $\R^d$, there are a variety of approaches that exploit coordinates to do so, most notably the work of Callahan and Kosaraju~\cite{callahan95decomposing}.
A variety of such techniques may be found in the book on geometric approximation algorithms by Har-Peled~\cite{har-peled11geometric}.
The finite Voronoi method provides a new way to avoid running times depending on the spread when doing metric divide-and-conquer in doubling spaces.
Although certainly not the first of its kind (see for example~\cite{Cole06Searching,har-peled06fast}), it is likely the simplest.
Implementations of these algorithms are available as a Python package (\texttt{pip install greedypermutations})~\cite{sheehy20greedypermutations}.

    \bibliographystyle{plain}
    \bibliography{bibliography}

    \appendix
    \section{Finite Voronoi Lemmas}\label{sec:voronoi_lemmas}

In this section, we state and prove lemmas about the structure of finite Voronoi diagrams and their neighbors graphs. (See Section~\ref{sec:fvm_and_gts}).
Specifically, we show that the Graph Update and Pruning sets maintain the Neighbor Invariant.
We also prove the degree bound for the neighbor graph when constructed in a greedy order.

\subsection{The Neighbor of Neighbors Graphs Update}\label{sec:nbr_of_nbrs}

\begin{figure}[htb]
    \begin{center}
    \begin{tikzpicture}[scale=1.5]
    \coordinate (A1) at (0,2);
    \coordinate (A2) at (1,1);
    \coordinate (A3) at (1.5,1);
    \coordinate (A4) at (2,2);

    \coordinate (A5) at (0.25,0.75);
    \coordinate (A6) at (0.75,1.5);

    \coordinate (A7) at (1.35,0.5);
    \coordinate (A8) at (1.25,1.5);

    \coordinate (A9) at (2.5,0.75);
    \coordinate (A10) at (2, 1.5);

    \draw[blur shadow={shadow blur steps=5,shadow blur extra rounding=1.5pt}] (A1) -- (A2) -- (A3) -- (A4) -- cycle;
    \shade[top color=yellow!50,bottom color=yellow!50] (A1) -- (A2) -- (A3) -- (A4) -- cycle;

    \filldraw (A5) circle (1pt);
    \filldraw (A6) circle (1pt);
    \draw[dashed, ->] (A5) -- (A6);

    \filldraw (A7) circle (1pt);
    \filldraw (A8) circle (1pt);
    \draw[dashed, ->] (A7) -- (A8);

    \filldraw (A9) circle (1pt);
    \filldraw (A10) circle (1pt);
    \draw[dashed, ->] (A9) -- (A10);

    \node [below, scale=0.8] at (A5) {$a$};
    \node [left, scale=0.8] at (A6) {$a_{'}$};
    \node [below, scale=0.8] at (A7) {$b$};
    \node [above, scale=0.8] at (A8) {$b_{'}$};
    \node [right, scale=0.8] at (A9) {$c$};
    \node [right, scale=0.8] at (A10) {$c_{'}$};
    \end{tikzpicture}
    \caption{Neighbor of neighbor update.  If inserting the point $a'$ from $\vor(a)$ creates an edge from $a$ to $c$, then there is a path $a$ to $b$ to $c$ in the neighbor graph.}
    \end{center}
\end{figure}
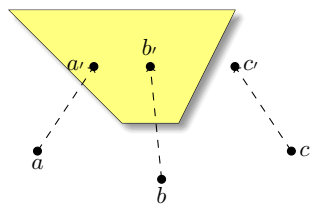

Recall that when inserting $a' \in \vor(a)$ as a new site in an FVD, we connect $a'$ to all sites within two steps of $a$ in the neighbor graph.
That is, when adding a new cell, we need only consider the neighbors of neighbors of the parent cell.

\begin{lemma}\label{lem:nbrs_of_nbrs}
Let $a'\in \vor(a)$.
Let $G$ be the neighbor graph before inserting $a'$ and let $G'$ be the graph after the insertion.
If the Neighbor Invariant requires an edge $a'c$ in $G'$, then there is a path of length at most two edges connecting $a$ to $c$ in $G$.
\end{lemma}
\begin{proof}
    Let $b'$ be a point that moved from $\vor(b)$ to $\vor(a')$ upon insertion.
    We know that $ab \in G$.
Let $c'\in \vor(c)$ be a point such that
\[
  \dist(b', c')
  \le \max\{\dist(b', a'), \dist(c', c)\} \le \max\{\dist(b', b), \dist(c', c)\},
\]
Then either $b=c$ or the Neighbor Invariant requires that $bc$ is an edge in $G$.
We therefore have the desired path $abc$ from $a$ to $c$.
\end{proof}

\subsection{Pruning Condition}\label{sec:pruning}
The pruning condition is a way to eliminate edges from the neighbor graph without checking explicitly.
The following lemma guarantees that any edge required by the Neighbor Invariant does not get pruned.

\begin{lemma}\label{lem:pruning}
Let $ab$ be an edge required by the Neighbor Invariant for a finite Voronoi diagram $V$.
    Then, $\dist(a,b) \le \outradius(a) + \outradius(b) + \max\{\outradius(a), \outradius(b)\}$.
\end{lemma}
\begin{proof}
    By the Neighbor Invariant, there exists points $a'\in \vor(a)$ and $b'\in \vor(b)$ such that $\dist(a', b')\le \max\{\dist(a', a), \dist(b',b)\}$.
We can then bound 
    \[
    \dist(a,b)\le \dist(a,a') + \dist(a', b') + \dist(b',b) 
    \le \outradius(a) + \outradius(b) + \max\{\outradius(a), \outradius(b)\}.\qedhere
\]
\end{proof}

    \subsection{Structural Properties of FVDs Constructed Using Clarkson's Algorithm}
\label{apx:fvd_properties}

In this section we prove properties of the FVDs constructed using $\CLARKSON(P, \lazy, \vorapx, \greedy)$ described in Section~\ref{sec:simple_clarkson}.
These structural properties are used to bound the running times of the algorithms in Section~\ref{sec:results}.

\begin{restatable}{lem}{degreebound}[Degree Bound Lemma]\label{lem:degree_bound}
    The number of neighbors of a cell in an FVD constructed running $\CLARKSON(P, \lazy, \vorapx, \greedy)$ is at most $(12\vorapx\greedy)^d$.
\end{restatable}
\begin{proof}
    Let $\emin$ be the minimum insertion distance among the sites inserted so far.
    Then, the sites are $(\emin/\vorapx)$-packed.
    By the Heap Invariant, the out-radius of every cell is at most $\greedy\emin$.
    The pruning condition guarantees that the distance from $p$ to any neighbor $q$ is bounded as follows.
    \[
        \dist(p,q) \le \outradius(p) + \outradius(q) + \max\{\outradius(p), \outradius(q)\}\le 3\greedy\emin.
    \]
    So, by the Standard Packing Lemma (Lemma~\ref{lem:std_packing}) the number of neighbors is at most $(12 \vorapx\greedy)^d$.
\end{proof}

\begin{lemma}\label{lem:cell_scaling}
    Let $p \in \vor(p')$ in an FVD constructed running $\CLARKSON(P, \lazy, \vorapx, \greedy)$.
    When $p$ is inserted as a site, the cell out-radius, $\outradius(p) \le \frac{\greedy}{\lazy} \dist(p,p')$.
\end{lemma}
\begin{proof}
    Suppose $q \in \vor(q')$ and it moved to $\vor(p)$ when $p$ was inserted from $\vor(p')$.
    It follows that,
    \begin{align*}
        \outradius(p) &\le \dist(p, q) & \because{definition of out-radius}\\
                        &\le \frac{1}{\lazy} \dist(q,q') & \because{$\lazy$-Lazy Moves} \\
                        &\le \frac{\greedy}{\lazy} \dist(p,p') &\because{by the $\greedy$-heap invariant}. &&\qedhere
    \end{align*}
\end{proof}

\begin{lemma}\label{lem:aspect_ratio_bound}[Aspect Ratio Lemma]
    The aspect ratio of a cell in an FVD constructed running $\CLARKSON(P, \lazy, \vorapx, \greedy)$ is at most $\frac{\vorapx\greedy(1+\vorapx)}{\lazy}$.
\end{lemma}
\begin{proof}    
    The out-radius of a cell is maximized when the point is first inserted.
    It will suffice to prove the aspect ratio is bounded for a site $p$ when it is first inserted and also every time the in-radius decreases.
    
    Let $p$ be any site.
    When $p$ is inserted from $\vor(p')$, let $q\in \vor(q')$ be the point such that $\inradius(p) = \dist(p,q)$.
    Then, after the insertion, we have
    \begin{align*}
            \outradius(p)
            &\le \frac{\greedy}{\lazy}\dist(p,p') & \because{Lemma~\ref{lem:cell_scaling}}\\
            &\le \frac{\vorapx\greedy}{\lazy}\dist(p, q') & \because{by the Cell Invariant}\\
            &\le \frac{\vorapx\greedy}{\lazy}(\dist(p,q) + \dist(q,q')) & \because{by the Triangle Inequality}\\
            &\le \frac{\vorapx\greedy(1+\vorapx)}{\lazy}\dist(p,q) & \because{by the Cell Invariant}\\
            &\le \frac{\vorapx\greedy(1+\vorapx)}{\lazy}\inradius(p). & \because{by the choice of $q$}\\
    \end{align*}
    
    The out-radius and in-radius can only decrease as more sites are inserted.
    Decreasing the out-radius will also decrease the aspect ratio, so we need not consider that case.
    If inserting some $q''\in \vor(q')$ causes a point $p''\in \vor(p)$ to move out of the cell so that $\inradius(p) = \dist(p,p'')$, then
    \begin{align*}
        \outradius(p)
            &\le \greedy\,\dist(q',q'')& \because{$\greedy$-heap invariant}\\
            &\le \vorapx\greedy\, \dist(p,q'')& \because{by the Cell Invariant}\\
            &\le \vorapx\greedy(\dist(p,p'') + \dist(p'',q''))& \because{by the Triangle Inequality}\\
            &\le \vorapx\greedy\left(1 + \frac{1}{\lazy}\right)\dist(p,p'')& \because{$\lazy$-Lazy Moves}\\
            &= \frac{\vorapx\greedy(1+\lazy)}{\lazy}\inradius(p). & \because{by the choice of $p''$}\\
            &\le \frac{\vorapx\greedy(1+\vorapx)}{\lazy}\inradius(p). & \because{because $\vorapx \ge 1$} && \qedhere
    \end{align*}
\end{proof}

We refer to $\aspect := \frac{\vorapx\greedy(1+\vorapx)}{\lazy}$ as the aspect ratio bound.

\begin{lemma}\label{lem:fvd_insertion_order}
    The insertion order of sites in $\CLARKSON(P, \lazy, \vorapx, \greedy)$ is $\vorapx\greedy$-greedy.
\end{lemma}
\begin{proof}
    Let $S$ denote the sites just before the insertion of $a \in \vor(a')$.
    It is sufficient to show that $\dist(p, S) \le \vorapx\greedy\dist(a, S)$ for any point $p$.

    Let $p \in \vor(p')$.
    We have,
    \begin{align*}
        \dist(p, S) &\le \dist(p, p')\\
                        &\le \outradius(p') &\because{definition of out-radii}\\
                        &\le \greedy \dist(a, a')  &\because{by the Heap Invariant} \\
                        &\le \greedy \vorapx \dist(a, S) & \because{the Cell Invariant}&&\qedhere
    \end{align*}
\end{proof}




    \section{The Structure Theorem for Greedy Trees}\label{sec:apx_structure}

In this section, we prove a structure theorem for $(\scale, \gpapprox, \locallygreedy)$-greedy trees.
In particular, we prove an upper bound on the radius of a greedy tree node.
We also provide a packing guarantee for nodes with disjoint leaves.

\begin{restatable}{restate}{structuretheorem}[Structure Theorem]\label{thm:structure}
    Let $G$ be an $(\scale, \gpapprox, \locallygreedy)$-greedy tree.
    \begin{enumerate}
        \item  \textbf{Radius Bounds}: For any node $x$ centered at $a$ with right child $y$ centered at $b$ we have,
            $\noderad{x} \le \min\{\frac {\e_a} {\scale-1}, \frac{\scale\locallygreedy\e_b}{\scale - 1}\}$.
        \item \textbf{Packing}: Let $x$ and $y$ be nodes with centers $a$ and $b$ respectively, such that $a \ne b$.
            Then $\gpapprox \dist(a,b) \ge \min\{\e_a,\e_b\}$.
        \item \textbf{Strong Packing}: For any node $x$ we have $\ball(\nodecenter{x},\,\pack\cdot\splitdist(x)) \subseteq \points(x)$, where $\pack = \frac{\scale^2-\scale-1}{(\scale^2-1)(\gpapprox+1)\gpapprox}$.
    \end{enumerate}
\end{restatable}
\begin{proof}
    \begin{enumerate}
        \item 
        By definition $\noderad{x} = \max_{c\in\points(x)} \dist(a,c)$.
        First we show that $\noderad{x} \le \frac{\e_a}{\scale-1}$ for any node by induction on the height of the subtree rooted at $x$.
        The base case is clear because the radius of a leaf is $0$.
        
        By the triangle inequality, $\noderad{x} \le \max\{ \noderad{x'},~\e_b + \noderad{y} \}$, where $x'$ and $y$ (with center $b$) are the left and right children of $x$ respectively.
        By the inductive hypothesis, we have, $\noderad{x} \le \max\{\frac{\e_a}{\scale-1}, \e_b + \frac{\e_b}{\scale-1}\}$.
        Then, $\noderad{x} \le \frac{\e_a}{\scale-1}$ follows from the $\scale$-scaling property of the greedy tree.

        Again by induction on the height we can conclude $\noderad{x} \le \max_{c \in \points(x)}\{\frac{\scale\e_c}{\scale-1}\}$ over the nodes in the left subtree of $x$.
        As the heap-order traversal of $G$ gives a $\locallygreedy$-greedy permutation,
        \[
            \max_{c \in \points(x)}\e_c \le \max_{c \in \points(x)}\dist(a,c) \le \locallygreedy\dist(a,b) = \locallygreedy\e_b.
        \]
        It follows that $\noderad{x} \le \frac{\locallygreedy \scale \e_b}{\scale-1}$.
        
        \item
        Let $P$ be the permutation used to construct $G$.
        Let $P_a$ and $P_b$ denote the prefixes of $a$ and $b$ in $P$ respectively.
        If $a \in P_b$ then,
        \[
            \e_b \le \gpapprox \dist(b, P_b) \le \gpapprox \dist(a,b).
        \]
        Similarly, if $b \in P_a$ then $\e_a \le \gpapprox\dist(a,b)$.
        Therefore, $\dist(a,b) \ge \frac{1}{\gpapprox}\min\{\e_a,\e_b\}$.
        
        \item See Theorem~\ref{thm:strong_packing}. \qedhere
    \end{enumerate}
\end{proof}

The following corollary follows from the radius bounds and packing of Theorem~\ref{thm:structure}.

\begin{corollary}
    \label{cor:structure}
    Let $G(\scale, \gpapprox, \locallygreedy)$ be a greedy tree.
    Let $X$ be a subset of nodes of $G$ with pairwise disjoint leaves.
    Let the radius of every parent of a node in $X$ be at least $r$.
    Then, the centers of the nodes in $X$ are $\frac{\scale-1}{\scale\gpapprox\locallygreedy}r$-packed.
\end{corollary}

    \section{Strong Packing in Greedy Trees}\label{sec:strong_packing}

The strong packing condition says that not only are independent nodes separated in terms of their radii (i.e., the normal packing condition), but also that this packing extends to their points.
A version of this condition was first shown for net-trees~\cite{har-peled06fast} and later for cover trees~\cite{jahanseir16transforming}.
More specifically, for any node $x$, we want that the points close to $\nodecenter{x}$ are in its subtree.
That is, for some constant $\pack$, we have
\[
    \ball(\nodecenter{x}, \pack\; \splitdist(x)) \subseteq \points(x)
\]

Let $a$ be a cell center and let $a'$ be the next point in $\vor(a)$ that is chosen to be the center of a new cell.
Let $x$ denote the node in the completed tree centered at $a$ whose right child is centered at $a'$.
The key idea of the proof is to show that as we are running the algorithm, we can look at the cell of $a$.
There is a radius such that the entire ball centered at that $a$ with that radius is contained inside the cell.
Then, after the algorithm has completed, some smaller ball centered at $a$ is a subset of $\points(x)$.

The pattern for the proof is to look at any point that encroaches the smaller circle and then move up the tree to its earliest ancestor that encroaches.
The lazy move constant can be used to show that this first encroachment cannot make too much progress towards the center, $a$.
Then, by scaling the subtree never makes it farther than some constant times closer.

\begin{theorem}\label{thm:strong_packing}
    Let $G$ be a greedy tree constructed using $\CLARKSON(P, \lazy, \vorapx, \greedy)$.
    Let $x$ be a node of $G$ centered at point $a$.
    Then $\ball(a, \pack\; \splitdist(x)) \subseteq \points(x)$, where 
    \[
        \pack = \frac{\lazy^2-\lazy-1}{(\lazy^2-1)(\vorapx+1)\vorapx}.
    \]
\end{theorem}
\begin{proof}
    Let $x$ be a node with center $a$ and sibling node $x'$ in the completed greedy tree $G$ obtained by running $\CLARKSON(P, \lazy, \vorapx, \greedy)$.
    If $x$ is the right child of its parent, then let $\vor(a)$ be the newly created cell centered at $a$ in the FVD.
    On the other hand, if $x$ is the left child, then let $\vor(a)$ be the Voronoi cell centered at $a$ just after $\nodecenter{x'}$ is inserted.

    By Lemma~\ref{lem:inrad_splitdist}, for either case, we have $\inradius(a) \ge \frac{1}{\vorapx(\vorapx+1)}\splitdist(x)$.
    So, it follows that  $\ball(a,\pack\;\splitdist(x))\subseteq \ball\left(a, \frac{\lazy^2-\lazy-1}{\lazy^2-1} \inradius(a)\right)$.

    To achieve our desired conclusion we now show that
\[
    \ball\left(a, \frac{\lazy^2-\lazy-1}{\lazy^2-1} \inradius(a)\right) \subseteq \points(x).
\]
    It is sufficient to show there are no non-descendents of $x$ within a distance $\left(1-\frac{\lazy}{\lazy^2-1}\right)\inradius(a)$.

Let $p$ be the nearest point to $a$ that is not in $\points(x)$.
Let node $y$ with center $b$ be the highest ancestor of $p$ such that $\dist(a, b) \le \inradius(a)$.
Let node $z$ with center $c$ be the parent of $y$.

By the definition of in-radius, at some time, the point $b$ moved out of $\vor(a)$ into $\vor(c)$.
So by the lazy move condition,
\[
    \dist(b,c) \le \frac{\dist(a,b)}{\lazy}.
\]

It follows that
\[
    \dist(a,b) \ge \dist(a,c) - \dist(b,c) \ge \inradius(a) - \frac{1}{\lazy}\dist(a,b).
\]
So,
\[
    \dist(a,b) \ge \frac{\lazy}{\lazy+ 1}\inradius(a).
\]

By the structure theorem, the node radius of $y$ can be bounded as follows.

\[
    \noderad{y} \le \frac{\dist(b,c)}{\lazy- 1} \le \frac{\dist(a,b)}{\lazy(\lazy-1)}.
\]

We can now bound the distance from $p$ to $a$ as follows.
\begin{align*}
    \dist(a,p)
        &\ge \dist(a,b) - \noderad{y}\\
        &\ge \dist(a,b) - \frac{\dist(a,b)}{\lazy(\lazy-1)}\\
        &\ge \left(1-\frac{1}{\lazy(\lazy-1)}\right)\frac{\lazy}{\lazy+ 1}\inradius(a)\\
        &= \left(1-\frac{\lazy}{\lazy^2-1}\right)\inradius(a). \qedhere
\end{align*}
\end{proof}

From this, we see that as long as $\lazy$ is greater than the golden ratio, we can expect to get some non-trivial strong packing.

We now prove the lemma used in the previous theorem.
\begin{lemma}\label{lem:inrad_splitdist}
    Let $\vor(p)$ be a newly inserted cell with predecessor $p'$ in an FVD when running $\CLARKSON(P, \lazy, \vorapx, \greedy)$.
    Then,
    \[
        \min\{\inradius(p'), \inradius(p)\} \ge \frac{1}{\vorapx(1+\vorapx)} \dist(p, p').
    \]
\end{lemma}
\begin{proof}
    
    \begin{enumerate}
        \item 
        First, we show that $\inradius(p) \ge \frac{1}{\vorapx(1+\vorapx)} \dist(p, p')$.
        Let $\inradius(p) = \dist(p, q)$ for $q \in \vor(q')$.
        \begin{align*}
            \dist(p, p')
                            &\le \vorapx\dist(p, q') &\because{Cell Invariant}\\
                            &\le \vorapx(\dist(p,q) + \dist(q',q)) &\because{Triangle Inequality} \\
                            &\le \vorapx(\dist(p,q) + \vorapx\dist(p,q)) &\because{Cell Invariant} \\
                            &= \vorapx(1+\vorapx)\inradius(p)
        \end{align*}
        \item Next, we show this is true for the parent cell $\vor(p')$.
        Let $\inradius(p') = \dist(p', p'')$ for $p'' \in \vor(p)$.
        \begin{align*}
            \inradius(p') &= \dist(p', p'') \\
                        &\ge \dist(p',p) - \dist(p, p'') \\
                        &\ge \dist(p',p) - \frac{1}{\lazy} \dist(p',p'') &\because{\textrm{$\lazy$-lazy PL}}.
        \end{align*}
        It follows that $\dist(p',p'') \ge \frac{\lazy}{1+\lazy} \dist(p, p')$.
        Therefore, by our choice of $\lazy$ and $\vorapx$,
        \[
            \inradius(p') \ge \frac{1}{\vorapx(1+\vorapx)} \dist(p,p').
        \]
        
        Let $\inradius(p') = \dist(p', \vor(q''))$ where $q'' \ne p$.
        The insertion of $q''$ caused some points to move from $\vor(p')$ to $\vor(q'')$.
        After the points moved to $\vor(q'')$, by the aspect ratio bound from Lemma~\ref{lem:aspect_ratio_bound} $\dist(p,p') \le \outradius(p') \le \aspect\inradius(p')$.
        Therefore,
        \[
            \inradius(p') \ge \frac{1}{\aspect} \dist(p,p') = \frac{\alpha}{\vorapx\greedy(1+\vorapx)}\dist(p,p') \ge\frac{1}{\vorapx(1+\vorapx)}\dist(p,p'). \qedhere
        \]
    \end{enumerate}
\end{proof}

    \subsection{Amortize Point Location Lemma}
\label{apx:amortize_pl}

\amortizepl*
\begin{proof}
    Let $X_r$ be the set of sites $p$ that touch $x$ such that $r\le \dist(\nodecenter{x}, p) < 2r$.
    Note that $X_r$ is $R$-packed, where $R$ is the minimum insertion distance of any point in $X_r$.
    This follows because in \REFINE predecessors are exact nearest neighbors.
    Then, by the Standard Packing Lemma, we have
        $|X_r| \le \left(8r/R\right)^d.$

    Let $b'$ be a site with $b = \pred(b')$.
    We now prove that if $x$ is touched by the insertion of $b'$ then $\dist(b,b') \ge \frac{r}{5\left(1+\frac{1}{n}\right)}$, which gives us a lower bound on the packing $R$.

    Let $V(a)$ be the cell containing $x$.
    So $a$ and $b$ are neighbors because, otherwise, $x$ would not have been touched.
    By the pruning condition,
        $\dist(a,b)\le \outradius(a) + \outradius(b) + \max\{\outradius(a), \outradius(b)\}$.
    \begin{align*}
        r   &\le \dist(\nodecenter{x},b') & \text{[by hypothesis]}\\
            &\le \dist(\nodecenter{x},a) + \dist(a,b) + \dist(b,b') & \text{[triangle inequality]}\\
            &\le \outradius(a) + (\outradius(a) + \outradius(b) + \max\{\outradius(a), \outradius(b)\}) + \outradius(b) & \text{[by Pruning and Radius Upper Bounds]}\\
            &\le 5\left(1 + \frac{1}{n}\right)\dist(b,b') & \text{[heap invariant]}
    \end{align*}
    Thus by the standard packing lemma $|X_r| \le \left(40\left(1+\frac{1}{n}\right)\right)^d$.
\end{proof}

    \subsection{Analyzing FVDs on Greedy Trees}\label{sec:tidying}

In this section we analyze some properties of FVDs produced by Clarkson's Algorithm when the input has been preprocessed into a greedy tree.
We first show that $\radapprox$-tidy cells satisfy the Radius Invariant.

A consequence of the tidying and split-on-move procedures is that the cells have constant complexity.
We use the following facts to show that the cells are sparse.
Their proofs are presented after the proof of Lemma~\ref{lem:constant_complexity}.

\begin{restatable}{lem}{tidypack}\label{lem:tidy_pack}
    Let $\vor(p)$ be a cell in an FVD when running $\CLARKSON(G, \lazy, \vorapx, \greedy)$, maintaining $\radapprox$-tidy cells.
    For any node $x \in \vor(p)$ such that the parent $x'$ of $x$ was split because of the tidying procedure,
    \[
        \noderad{x'} > \frac{(\radapprox-1)\lazy}{\radapprox^2\aspect(\lazy+1)}\outradius(p).
    \]
\end{restatable}
\begin{restatable}{lem}{splitonmovepack}\label{lem:splitonmove_pack}
    Let $\vor(p)$ be a cell in an FVD when running $\CLARKSON(G, \lazy, \vorapx, \greedy)$, maintaining $\radapprox$-tidy cells.
    For any node $x \in \vor(p)$ such that the parent $x'$ of $x$ was split-on-move,
    \[
        \noderad{x'} > \frac{\vorapx-\lazy}{2\vorapx\greedy\aspect(\lazy+1)(\vorapx+1)}\outradius(p).
    \]
\end{restatable}

Here, $\aspect$ is the aspect ratio bound from Lemma~\ref{lem:aspect_ratio_bound}.
Now, we prove Lemma~\ref{lem:constant_complexity}.
\constcellcomplexity*
\begin{proof}
    Let $\vor(p)$ be a cell in an FVD.
    For node $x \in \vor(p)$, let $x'$ be its parent node.
    Either $x'$ was split-on-move, or it was split when the cell containing $x'$ was tidied.
    If $x'$ was split when tidying, then by Lemma~\ref{lem:tidy_pack},
    \[
        \noderad{x'} > \frac{(\radapprox-1)\lazy}{\radapprox^2\aspect(\lazy+1)}\outradius(p).
    \]
    On the other hand, if $x'$ was split when deciding to move it then by Lemma~\ref{lem:splitonmove_pack},
    \[
        \noderad{x'} > \frac{\vorapx-\lazy}{2\vorapx\greedy\aspect(\lazy+1)(\vorapx+1)}\outradius(p).
    \]
    Let $\zeta$ be such that,
    \[
        \zeta := \max\bigg\{\frac{2\vorapx\greedy\aspect(\lazy+1)(\vorapx+1)}{\vorapx-\lazy}, \frac{\radapprox^2\aspect(\lazy+1)}{(\radapprox-1)\lazy}\bigg\}.
    \]

    
    Let $G$ be an $(\scale', \gpapprox', \locallygreedy')$-greedy tree.
    If nodes $y$ and $z$ are in $\vor(p)$, then by Lemma~\ref{cor:structure},
    \[
        \dist(\nodecenter{y},\nodecenter{z}) \ge \frac{\scale'-1}{\scale'\gpapprox'\locallygreedy'}\frac{\outradius(p)}{\zeta}.
    \]
    Therefore, by Lemma~\ref{lem:std_packing} the complexity of $\vor(p)$ is $\zeta^{O(d)}$.
\end{proof}

Next, we prove Lemma~\ref{lem:tidy_pack}.

\tidypack*
\begin{proof}
    There are two cases to consider.
    Either $x'$ was split when $\vor(p)$ was tidied or it had already been split before $\vor(p)$ was inserted.
    \begin{enumerate}
        \item
        Let $x'\in \vor(p)$ when it was split.
        Because $x'$ was split as part of the tidying procedure,
        \[
            \noderad{x'} > (\radapprox-1)\dist(p,\nodecenter{x'}),
        \]
        and,
        \[
            \dist(p, \nodecenter{x'}) + \noderad{x'} \ge \outradius'(p),
        \]
        where $\outradius'(p)$ is the out-radius of $\vor(p)$ when $x'$ was split.
        As outradii are non-increasing, it follows that $\noderad{x'} > \frac{\radapprox-1}{\radapprox}\outradius(p)$.
        \item 
        Let $x' \in \vor(q)$ when it was split where $q \ne p$.
        Let $\outradius'(q)$ be the out-radius of $\vor(q)$ when $x'$ was split.
        As $x \in \vor(q)$, by a similar argument as the previous case,
        \[
            \dist(q, \nodecenter{x}) \le \outradius'(q) < \frac{\radapprox}{\radapprox-1}\noderad{x'}.
        \]

        Suppose $x \in \vor(p')$ when $\vor(p)$ was inserted.
        By $\lazy$-lazy moves, we have $\lazy\dist(p, \nodecenter{x}) \le \dist(p', \nodecenter{x})$.
        It follows that,
        \begin{align*}
            \dist(p,p') &\le \dist(p,\nodecenter{x}) + \dist(p', \nodecenter{x}) \because{triangle inequality}\\
                        &\le \frac{\lazy+1}{\lazy} \dist(p', \nodecenter{x}) \because{$\lazy\dist(p, \nodecenter{x}) \le \dist(p', \nodecenter{x})$}\\
                        &\le \frac{\lazy+1}{\lazy} \dist(q, \nodecenter{x}) \because{$\dist(p', \nodecenter{x}) \le \dist(q, \nodecenter{x})$}\\
                        &< \frac{\lazy+1}{\lazy}\cdot\frac{\radapprox}{\radapprox-1} \noderad{x'}.
        \end{align*}
    
        Therefore, by the definition of the in-radius and Lemma~\ref{lem:aspect_ratio_bound},
        \[
            \noderad{x'} > \frac{\radapprox-1}{\radapprox}\cdot\frac{\lazy}{\lazy+1}\inradius(p) \ge \frac{\radapprox-1}{\radapprox^2}\cdot\frac{\lazy}{\lazy+1}\cdot\frac{1}{\aspect}\outradius(p). \qedhere
        \]
    \end{enumerate}
\end{proof}

The following lemma is a simple consequence of the split-on-move procedure.
We use it to prove Lemma~\ref{lem:splitonmove_pack}.

\begin{restatable}{lem}{splitonmove}\label{lem:splitonmove}
    Let $\vor(p)$ be a newly created cell in an FVD when running $\CLARKSON(G, \lazy, \vorapx, \greedy)$.
    Let node $x \in \vor(p')$ such that $x$ was split-on-move during the insertion of $\vor(p)$.
    Then,
    \[
        \noderad{x} > \frac{\vorapx-\lazy}{2(\lazy+1)(\vorapx+1)}\dist(p, p').
    \]
\end{restatable}

\splitonmovepack*
\begin{proof}
    There are two cases to be considered.
    Either $x'$ was split before $\vor(p)$ was inserted, or $x' \in \vor(p)$ and the insertion of some cell in the vicinity of $\vor(p)$ caused $x'$ to split but then $x$ did not move into the new cell.
    \begin{enumerate}
        \item
        Suppose $x$ moved into $\vor(p)$ when $\vor(p)$ was created.
        Let $\vor(q)$ be the last cell to contain $x'$.
        There exists some cell $\vor(q')$ such that the insertion of $\vor(q')$ caused $x'$ to split.
        Moreover, by Lemma~\ref{lem:fvd_insertion_order}, insertions are $\vorapx\greedy$-greedy.
        It follows that,
        \begin{align*}
            \noderad{x'} &> \frac{\vorapx-\lazy}{2(\lazy+1)(\vorapx+1)}\dist(q, q') &\because{Lemma~\ref{lem:splitonmove}}\\
                        &\ge \frac{1}{\vorapx}\cdot\frac{\vorapx-\lazy}{2(\lazy+1)(\vorapx+1)}\e_{q'} &\because{$\vorapx$-approximate cells} \\
                        &\ge \frac{1}{\vorapx\greedy}\cdot\frac{\vorapx-\lazy}{2(\lazy+1)(\vorapx+1)}\e_{p} &\because{$\vorapx\greedy$-greedy insertions} \\
                        &\ge \frac{1}{\vorapx\greedy}\cdot\frac{\vorapx-\lazy}{2(\lazy+1)(\vorapx+1)}\inradius(p) &\because{definition of in-radius}.
        \end{align*}
        \item As $x' \in \vor(p)$ and $x'$ was split-on-move, there exists some cell $\vor(p')$ such that the insertion of $\vor(p')$ caused $x'$ to split.
        It follows that,
        \begin{align*}
            \noderad{x'} &> \frac{\vorapx-\lazy}{2(\lazy+1)(\vorapx+1)}\dist(p, p') &\because{Lemma~\ref{lem:splitonmove}}\\
                        &\ge \frac{\vorapx-\lazy}{2(\lazy+1)(\vorapx+1)}\inradius(p) &\because{definition of in-radius} \\
                        &\ge \frac{1}{\vorapx\greedy}\cdot\frac{\vorapx-\lazy}{2(\lazy+1)(\vorapx+1)}\inradius(p). &\because{$\vorapx, \greedy\ge 1$}
        \end{align*}
    \end{enumerate}
    The inequality follows from Lemma~\ref{lem:aspect_ratio_bound} in both cases.
\end{proof}

Finally, we prove Lemma~\ref{lem:splitonmove}.
\splitonmove*
\begin{proof}
    Because $x$ was split-on-move when $p$ was inserted, it is true that,
    \[
        \dist(p',\nodecenter{x})-\noderad{x} < \lazy(\dist(p,\nodecenter{x})+\noderad{x}),
    \]
    and,
    \[
        \vorapx(\dist(p,\nodecenter{x}) - \noderad{x}) < \dist(p',\nodecenter{x})+\noderad{x}.
    \]
    Thus, by the preceding two statements, we have
    \[
        \dist(p',\nodecenter{x}) < \frac{\lazy(\vorapx+1)+\vorapx(\lazy+1)}{\vorapx-\lazy}\noderad{x},
    \]
    and,
    \[
        \dist(p,\nodecenter{x}) < \frac{\lazy+\vorapx+2}{\vorapx-\lazy}\noderad{x}.
    \]
    Therefore,
    \[
        \dist(p,p') \le \dist(p,\nodecenter{x}) + \dist(p',\nodecenter{x}) <  \frac{2(\lazy+1)(\vorapx+1)}{\vorapx-\lazy}\noderad{x}. \qedhere
    \]
\end{proof}

    \subsection{Backburner Analysis}
\label{sec:apx_bb}
Let $\CLARKSONBB(G, \lazy, \vorapx, \greedy)$ refer to running $\CLARKSON(G, \lazy, \vorapx, \greedy)$ where the main heap is a $\bucketsize = \sqrt{\greedy}$-bucket queue with a backburner.
The following lemma shows that if a cell is placed onto the backburner there must be an empty annulus around that cell containing no points from any other cell.

\bucketqueue*
\begin{proof}
    Let $\rmax$ be the largest outradius on the bucket queue.
    If $\vor(p)$ is moved to the backburner and there are $s$ buckets, it follows that $\outradius(p) < \frac{\rmax}{\bucketsize^s}$.
    Suppose for contradiction that a point $a$ exists in a cell $\vor(q)$ and $\outradius(p) < \dist(p,a) \le \frac{\bucketsize^{s-1}}{\vorapx(1+\vorapx)}\outradius(p)$.
    Because $\dist(a,p) > \radius(p)$, we know that $a$ is not in $\vor(p)$.
    It follows that,
    \begin{align*}
        \dist(p,a) &= \frac{1}{1+\vorapx}(\dist(p,a) + \vorapx\dist(p,a)) \\
                    &\ge \frac{1}{1+\vorapx}(\dist(p,a) + \dist(q,a)) &\because{Cell Invariant}\\
                    &\ge \frac{1}{1+\vorapx}\dist(p,q) &\because{Triangle Inequality}\\
                    &\ge \frac{\rmax}{\vorapx(\vorapx+1)\bucketsize} &\because{Sites are $\frac{\rmax}{\vorapx\bucketsize}$-packed by Theorem~\ref{thm:structure}}\\
                    &> \frac {\bucketsize^{s-1}}{\vorapx(1+\vorapx)}\outradius(p) &\because{Backburner Condition}.
    \end{align*}
    However, this contradicts our assumption that $\dist(p,a) \le \frac{\bucketsize^{s-1}}{\vorapx(1+\vorapx)}t$.
\end{proof}

We refer to $\annulus = \frac{\bucketsize^{s-1}}{\vorapx(1+\vorapx)}$ as the empty annulus constant.
Next, we show that when a cell is removed from the backburner, it has no neighbors in the neighbor graph.


\bbisolation*
\begin{proof}
    We first show that if a cell is on the backburner, then its radius is less than the radius of its neighbors.
    Let $\vor(p)$ be a cell on the backburner and let $\vor(q)$ be its neighbor.
    By Lemma~\ref{lem:bucket_queue}, $\dist(p,q) > \annulus \outradius(p)$.
    By the pruning condition, $\dist(p,q) \le \outradius(p) + \outradius(q) + \max\{\outradius(p), \outradius(q)\}$.
    By our choice of $\annulus$, it follows that $\outradius(q) + \max\{\outradius(p), \outradius(q)\} > 2\outradius(p)$.
    Therefore, $\outradius(q) > \outradius(p)$.

    The main queue is empty when $\vor(p)$ is removed from the backburner.
    All the cells are on the backburner and it follows that they are isolated.
\end{proof}


Next, we formally prove that the greedy tree constructed with a bucket queue is the same regardless of whether or not a backburner was used.
We will use the following lemma.

\covering*
\begin{proof}
    \CLARKSONBB uses a $\radapprox$-bucket queue and cell outradii are $\radapprox$-approximate, where $\radapprox^2 \le \greedy$.
    So, if $a$ precedes $b$ in the heap-order traversal of $G$, then $\e_b \le \greedy\e_a$.
    Let $p \not \in P_a$.

    First, we consider the case where $\pred(p) \in P_a$.
    It follows that,
\[
        \dist(p, P_a) \le \dist(p, \pred(p)) 
                        = \e_p 
                        \le \greedy \e_a.
\]

    Now, suppose $\pred(p) \not \in P_a$.
    Suppose $p' \in P_a$ and $q$ is the first point not in $P_a$ such that $p$ moved from $\vor(p')$ to $\vor(q)$ in \CLARKSONBB.
    Thus,
\[
        \dist(p, P_a) \le \dist(p, p') 
                    \le \greedy \e_q 
                    \le \greedy^2 \e_a. \qedhere\]
\end{proof}

In fact, using a more careful and complex analysis it can be shown that $\dist(p, P_a) \le \greedy \e_a$ for all $p \in P$.
Now, we show that the order in which the cells are removed from the backburner does not affect the constructed greedy tree.

    





\end{document}